\documentclass[11pt]{amsart}

\usepackage{amsmath,amssymb,latexsym,amsthm,enumerate}
\usepackage{verbatim}
\usepackage{enumerate} 
\usepackage{hyperref}

\usepackage[dvips]{graphicx}
\usepackage{subfigure}

\theoremstyle{plain}
\newtheorem{theorem}{Theorem}
\newtheorem{proposition}[theorem]{Proposition}
\newtheorem{lemma}[theorem]{Lemma}
\newtheorem{corollary}[theorem]{Corollary}
\theoremstyle{definition}
\newtheorem{definition}[theorem]{Definition}

\theoremstyle{remark}

\newtheorem*{remark}{Remark}
\newtheorem*{remarks}{Remarks}

\newcounter{numpar}[section]

\newcommand*{\wt}{\widetilde}
\newcommand*{\wh}{\widehat}

\newcommand*{\dd}{\mathrm d}     

\newcommand*{\cD}{\mathcal D}

\newcommand*{\cF}{\mathcal F}

\newcommand*{\cI}{\mathcal I}
\newcommand*{\cJ}{\mathcal J}

\newcommand*{\bbC}{\mathbb C}
\newcommand*{\bbN}{\mathbb N}

\newcommand*{\bbR}{\mathbb R}

\newcommand*{\EE}{\mathsf E}
\newcommand*{\PP}{\mathsf P}

\DeclareMathOperator{\sgn}{sgn}

\newcommand*{\E}{\mathrm{e}}
\newcommand{\Rplus}{\mathbb{R}_{\geqslant 0}}
\newcommand{\pd}[2]{\frac{\partial #1}{\partial #2}}

\newcommand{\exit}{{\mbox{\, \vspace{3mm}}}\hfill\mbox{$\square$}}

\begin{document}
\title[Large deviations in affine stochastic volatility models with jumps]{Large 
deviations and stochastic volatility with jumps: 
\\asymptotic implied volatility for affine models}

\author{Antoine Jacquier}
\address{Institut f\"ur Mathematik - Technische Universit\"at Berlin, Germany}
\email{jacquier@math.tu-berlin.de}

\author{Martin Keller-Ressel}
\address{Institut f\"ur Mathematik - Technische Universit\"at Berlin, Germany}
\email{mkeller@math.tu-berlin.de}

\author{Aleksandar Mijatovi\'{c}}
\address{Department of Statistics, University of Warwick, UK}
\email{a.mijatovic@warwick.ac.uk}



\thanks{AJ would like to thank MATHEON for financial support.}

\begin{abstract} 
Let
$\sigma_t(x)$
denote the implied volatility
at maturity
$t$
for a strike
$K=S_0\E^{xt}$,
where
$x\in\bbR$
and
$S_0$
is the current value of the underlying.
We show that
$\sigma_t(x)$
has a \textit{uniform} (in $x$) limit as maturity
$t$
tends to infinity, given by the formula
$\sigma_\infty(x)  =  \sqrt{2}\left(h^*(x)^{1/2}+\left(h^*(x)-x\right)^{1/2}\right)$,
for 
$x$
in some compact neighbourhood of
zero
in the class of affine stochastic volatility models.
The function 
$h^*$
is the convex dual
of  
the limiting cumulant generating function 
$h$
of the scaled log-spot process. We express 
$h$
in terms of the functional characteristics 
of the underlying model.
The proof of the limiting formula rests on the 
large deviation behaviour of the scaled log-spot process 
as time tends to infinity.
We apply our results to obtain the limiting smile for
several classes of stochastic volatility models
with jumps used in applications (e.g. Heston with state-independent jumps,
Bates with state-dependent jumps and Barndorff-Nielsen-Shephard model).
\end{abstract}

\maketitle

\section{Introduction}
\label{sec:intro} 
Let
the process
$S=\E^X$
model
a risky security under an
equivalent martingale measure
and let
$\sigma_t(x)$
denote the implied volatility
at maturity
$t$
for a strike
$K=S_0\E^{xt}$
(see~\eqref{eq:Simga_t_Def}
for the precise definition of
$\sigma_t(x)$).
The main result of the present paper
(Theorem~\ref{thm:AsymVol})
states that, if the log-spot
$X$
follows an affine stochastic volatility process with jumps,
then $\sigma_t(x)$ converges to $\sigma_\infty(x)$ as the maturity $t$ tends to infinity,
where $\sigma_\infty$ is given by the formula
\begin{eqnarray}
\label{eq:FirstEq}
\sigma_\infty(x) & = & \sqrt{2}\left(h^*(x)^{1/2}+\left(h^*(x)-x\right)^{1/2}\right).
\end{eqnarray}
The function
$h$
is the limiting cumulant generating function of the scaled 
log-spot
$(X_t/t)_{t\geq1}$
and
$h^*$
is its convex dual (i.e. the Fenchel-Legendre transform of $h$).
Locally uniform convergence 
of the implied volatility to $\sigma_\infty$ is also established.

In~\cite{FJ,FJM_note} the limiting behaviour of
the smile at large maturities in the Heston model
is investigated.
Theorem~\ref{thm:AsymVol} can be viewed as a generalisation
of the main result in~\cite{FJ,FJM_note}. Not only does it cover a large 
class of stochastic volatility models with jumps rather than a 
single affine model with continuous trajectories, but furthermore
provides a better understanding of the limit:
Theorem~\ref{thm:AsymVol}
states that the limit holds also at the critical points
$x^*$
and
$\widetilde{x}^*$,
which are excluded from the analysis in~\cite{FJ, FJM_note}, 
and the convergence on the set
$\bbR\setminus\{x^*,\widetilde{x}^*\}$
is shown to be uniform on compact subsets.

In the class of affine stochastic volatility models,
the formula for the limiting implied volatility for
a fixed strike proved in Tehranchi~\cite{Tehr09} (see also~\cite{Lewis00} 
in the case of the Heston model)
also follows from~\eqref{eq:FirstEq}, since in Theorem~\ref{thm:AsymVol}
the convergence is uniform on a 
compact neighbourhood of the origin.
In~\cite{GaoLee}, the authors give 
various representations
for the implied volatility, 
including in the large-maturity regime,
based on an assumed asymptotic behaviour of certain European derivatives
in the underlying model, which is not specified.
This representation is not fully explicit in terms of the model parameters 
and it is therefore unclear how to apply it directly to the class 
of affine stochastic volatility models.

Contribution of the paper is twofold. 
First we study the properties
of the limiting cumulant generating function
$h$
of the affine stochastic volatility models.
Results in
Lemma~\ref{lem:h_not_zero}, Theorem~\ref{thm:LCGFh}
and Corollary~\ref{cor:LCGFh} give new properties of 
the function
$h$,
which are crucial for the understanding of the large
deviation behaviour of the model.
Second, the problem
of understanding the limiting behaviour of option
prices and the corresponding implied volatilities
using the large deviation principle 
is tackled.
The uniform limit in $x$
(on all compact subsets of $\bbR$)
of vanilla option prices is given in Theorem~\ref{thm:OptionPrices}
for non-degenerate affine stochastic volatility models
and 
exponential L\'evy models
(i.e. degenerate affine stochastic volatility models).
As mentioned above
Theorem~\ref{thm:AsymVol}
deals with the limiting implied volatility
smiles in these classes of models.

Besides giving a formula, which relates model parameters
and the limiting implied volatility smile, 
these theoretical results have the following practical consequences:\\
\noindent (1) in the large-maturity regime studied in this paper, the jumps in the model
influence the limiting implied volatility
smile as maturity tends to infinity (see examples in~\ref{subsec:PlotsOfSmiles}); \\
\noindent (2) for every affine stochastic volatility model there exists an exponential L\'evy model
such that the smiles of the two models in the limit coincide. In other words the stochasticity of
volatility does not (in the affine class) enlarge the family of possible limiting implied volatility
smiles (see Section~\ref{sec:App} for details).

The starting point of the analysis of the large deviation 
behaviour of an affine stochastic volatility 
process
$(X,V)$
in the present paper is 
Theorem~\ref{Thm:wm_convergence},
taken from~\cite[Theorem~3.4]{K2008a}.
This result
describes certain properties of
the limiting cumulant generating function
$h$,
which are however 
insufficient to understand the essential smoothness of
$h$
required in 
establishing the large deviation
principle of 
$(X_t/t)_{t\geq1}$.
The main contribution of this paper in the area of 
affine processes is 
Theorem~\ref{thm:LCGFh},
which identifies sufficient conditions for the process
$(X,V)$
that imply essential smoothness of the function 
$h$.
The conditions in
Theorem~\ref{thm:LCGFh}
are easy to apply to the models of interest (see e.g. Section~\ref{subsec:Examples}).
Its proof goes beyond the analysis 
in~\cite{K2008a}
as 
one is forced to study 
the special L\'evy-Khintchine form of
the characteristics
of the process,
since their general convexity properties 
no longer suffice to establish the required behaviour of the limit.



The rest of the paper is organized as follows.
In Section 2 we define the class of affine stochastic volatility processes and recall some of 
their properties.  In Section 3 we review briefly basic concepts in the theory of large deviations 
and state the G\"artner-Ellis theorem. Section 4 establishes the large deviations principle for the
scaled log-stock of an affine stochastic volatility model as maturity tends to infinity.
Sections 5 and 6 respectively translate this result into option price and implied volatility 
asymptotics. Numerical examples are given at the end of Section~6.

\section{Affine stochastic volatility models with jumps}
\label{sec:Affine}
Consider a stochastic model 
for a risky security 
$S=(S_t)_{t \geq 0}$
given by
\begin{eqnarray}
\label{eq:MainStockModel}
S_t & = &
\exp((r-d)t + X_t),\quad t \geq 0\;,
\end{eqnarray}
where the interest rate $r$ 
and the dividend yield $d$
are non-negative and constant and 
the log-price process $X=(X_t)_{t\geq0}$ 
starts at
$X_0\in \bbR$. 
Since the dynamics of 
$S$
is given under a risk-neutral measure,
the forward price process is 
$(\exp(X_t))_{t\geq0}$. 
We assume throughout the paper without loss of generality 
that $S$ is a forward price process (i.e. $r = d$).
Denote by $V=(V_t)_{t \geq 0}$ a process, starting at a constant
level 
$V_0 > 0$. 
The process
$V$
can be interpreted as the instantaneous variance process
of $X$ but may also control the arrival rate of jumps
of $X$. 
We make the following assumptions 
on the process
$(X,V)$
throughout the paper. 
\begin{description} 
\item[A1] $(X,V)$ is a stochastically continuous,
time-homogeneous Markov process with state-space $D = \bbR \times \Rplus$, 
where
$\Rplus:=[0,\infty)$.
\label{Eq:main_ass_sub1}
\item[A2] The cumulant generating function $\Phi_t(u,w)$ of
$(X_t,V_t)$ is of a particular affine form: there
exist functions $\phi(t,u,w)$ and $\psi(t,u,w)$ such that
\begin{eqnarray*}\label{Eq:affine}
\Phi_t(u,w) & := & \log \EE\left[\left.\exp(u X_t + w
V_t)\right|X_0,V_0\right]\\
& = & \phi(t,u,w) + V_0 \psi(t,u,w) + X_0 u
\end{eqnarray*}
for all $(t,u,w) \in \Rplus \times \bbC^2$, where the expectation
exists.\label{Eq:main_ass_sub2}
\end{description}

\begin{remarks}
\noindent (i) A1 and A2 make $(X,V)$ into an affine process in the sense
of~\cite{Schachermayer}.


\noindent (ii) 
A1 and A2 imply a homogeneity property of $(X,V)$:
if the starting value $(X_0,V_0)$ is shifted by
$(x,0)\in D$, 
the law of the random variable 
$(X_t,V_t)$ is shifted by the vector $(x,0)$ for any $t \ge 0$.

\noindent (iii) 
Assumptions A1 and A2 imply that the variance process
$V$ is a one-dimensional strong Markov process in its own right. 

\noindent (iv) The law of iterated expectations applied to
$\Phi_t(u,w)$ yields the flow-equations for $\phi$ and $\psi$ 
(see 
\cite[Eq.~(3.8)--(3.9)]{Schachermayer}):
\begin{equation}\label{Eq:flow_prop}
\begin{split}
\phi(t+s,u,w) &= \phi(t,u,w) + \phi(s,u,\psi(t,u,w)),\\
\psi(t+s,u,w) &= \psi(s,u,\psi(t,u,w)),
\end{split}
\end{equation}
for all $t, s \ge 0$. 

\noindent (v) It is
shown
in~\cite[Thm.~2.1]{K2008a}
(see
also~\cite{Schachermayer})
that if
$|\phi(\tau,u,\eta)|, |\psi(\tau,u,\eta)|< \infty$ 
for $(\tau,u,\eta) \in (0,\infty) \times \bbC^2$, then 
for all $t \in [0,\tau)$ and $w \in \bbC$ such that
$\Re\,w \le \Re\,\eta$,
the functions $\phi$ and $\psi$ satisfy
the generalized Riccati equations
\begin{subequations}\label{Eq:gen_Riccati}
\begin{align}
\partial_t \phi(t,u,w) &= F(u,\psi(t,u,w)), \quad \phi(0,u,w) = 0, \label{Eq:gen_Riccati_sub1}\\
\partial_t \psi(t,u,w) &= R(u,\psi(t,u,w)), \quad \psi(0,u,w) = w, \label{Eq:gen_Riccati_sub2}
\end{align}
\end{subequations}
where
\begin{equation} 
\label{eq:F_and_R}
F(u,w) :=
\left.\pd{}{t}\phi(t,u,w)\right|_{t = 0+}, \qquad R(u,w) :=
\left.\pd{}{t}\psi(t,u,w)\right|_{t = 0+}.
\end{equation}
Furthermore for all $t\in[0,\tau]$ 
we have
$|\phi(t,u,w)|,|\psi(t,u,w)| < \infty$.

\noindent (vi) If 
$(X,V)$
is a diffusion process,
then ODEs~\eqref{Eq:gen_Riccati}
become classical Riccati. Note also that~\eqref{Eq:gen_Riccati}
follows from the flow equations~\eqref{Eq:flow_prop}
by differentiation with respect to 
$s$.

\noindent (vii) $\phi$ and $\psi$ can for small 
$t$
be expressed implicitly in terms of 
$F$
and
$R$
as
\begin{equation*}
\phi(t,u,w) = \int_0^t{F(u,\psi(s,u,w))\;\dd s}\qquad\text{and}\qquad
    \int_w^{\psi(t,u,w)}{\frac{\dd\eta}{R(u,\eta)}} = t\;.
\end{equation*}
\end{remarks}

The functions 
$F$ and $R$, defined in~\eqref{eq:F_and_R}, must be of
L\'evy-Khintchine form
(see~\cite{Schachermayer}). In other words
\begin{subequations}\label{Eq:FR_form}
\begin{align}
\label{Eq:F_form}
F(u,w) &= \left\langle \frac{a}{2}(u,w)', (u,w)' \right\rangle +
\left\langle b, (u,w)' \right\rangle -c \\
&+\int_{D \setminus \{0\}} {\left(\E^{\langle\xi,(u,w)'\rangle} - 1 - \langle
\omega_F(\xi), (u,w)' \rangle \right)\,m(\dd\xi)} \notag,\\
R(u,w) &= \left\langle \frac{\alpha}{2}(u,w)', (u,w)' \right\rangle +
\left\langle \beta, (u,w)' \right\rangle -\gamma \\
&+\int_{D \setminus \{0\}}{\left(\E^{\langle\xi,(u,w)'\rangle} - 1 - \langle
\omega_R(\xi), (u,w)' \rangle \right)\,\mu(\dd\xi)} \notag,
\end{align}
\end{subequations}
where $D = \bbR \times \Rplus$,
$\langle\cdot,\cdot\rangle$
is the inner product on 
$\bbR^2$,
$(u,w)'$
denotes transposition,
$\omega_F$, $\omega_R$ are
suitable truncation functions, which we fix by defining
\begin{equation*}
\omega_F(\xi) = \left(\begin{array}{@{}c@{}}\frac{\xi_1}{1+\xi_1^2}\\0
\end{array}\right) \qquad \text{and} \qquad
\omega_R(\xi) =
\left(\begin{array}{@{}c@{}}\frac{\xi_1}{1+\xi_1^2}\\\frac{\xi_2}{1+\xi_2^2}
\end{array}\right),\quad\text{where} \quad
\xi =
\left(\begin{array}{@{}c@{}}\xi_1\\ \xi_2
\end{array}\right),
\end{equation*}
and the parameters $(a,\alpha,b,\beta,m,\mu)$
satisfy the following admissibility conditions:
\begin{itemize}
\item $a, \alpha$ are positive semi-definite 
$2\times2$-matrices with 
$a_{12} = a_{21} = a_{22} = 0$; 
\item $b \in D$, $\beta \in \bbR^2$ and $c,\gamma\in\Rplus$;  
\item $m$ and $\mu$ are L\'evy measures on $D$ and $\int_{D \setminus
\{0\}}{\left((\xi_1^2 + \xi_2) \wedge 1\right)\,m(\dd \xi)} < \infty$.
\end{itemize}

Assumptions A1 and A2, the
generalized Riccati equations and the L\'evy-Khintchine
decomposition~\eqref{Eq:FR_form} 
lead to the following interpretation of $F$ and $R$: 
$F$ characterizes the state-independent dynamics of the 
process $(X,V)$ 
while 
$R$
characterizes its state-dependent dynamics. 
The instantaneous characteristics of
the Markov process 
$(X,V)$
are given as follows:
$a + V\alpha $ the instantaneous covariance matrix, 
$b + V \beta$ the instantaneous drift, 
$m(\dd\xi) + V \mu(\dd\xi)$
the instantaneous arrival rate of jumps with jump heights in 
$\dd\xi$ and $c + \gamma V$ the instantaneous
killing rate.

The function $\chi$
defined below  
plays a key role in the characterisation of the 
martingale property of the process
$S=\exp(X)$.

\begin{definition}
\label{def:chi}
For each $u \in \bbR$ such that $R(u,0) < \infty$, define $\chi(u)$ as
\[\chi(u) := \left.\partial_2 R(u,w)\right|_{w = 0} := \left.\pd{R}{w}(u,w)\right|_{w = 0}\;.\]
\end{definition}

\begin{remarks}
\noindent (i) The condition $R(u,0) < \infty$
implies 
that, 
for some
$\delta>0$
the function $w\mapsto R(u,w)$
is 
convex on
$(-\delta,0]$ 
and differentiable 
on 
$(-\delta,0)$,
since the process
$V$
does not have negative jumps.
Therefore
$\chi(u)$ 
is a well-defined, possibly equal to 
$+\infty$, 
convex function 
given by the limit 
of 
$\partial_2 R(u,w)$
as $w \uparrow 0$.
It can be expressed explicitly
as
\[\chi(u) = \alpha_{12} u + \beta_1 + \int_{D \setminus \{0\}}
{\xi_2 \left(\E^{u\xi_1} - \frac{1}{1 + \xi_2^2}\right)\,\mu(\dd\xi)},\qquad
\text{where}\quad \xi'=(\xi_1,\xi_2)\;.\]

\noindent (ii) The sufficient and necessary condition for 
$S$
to be conservative and a martingale, in terms of 
$R,F$
and
$\chi$,
is given in~\cite[Thm.~2.5]{K2008a}. A simple sufficient
condition for these properties reads (see~\cite[Cor. 2.7]{K2008a}):
\begin{enumerate}
\item[$\bullet$] if $F(0,0) = R(0,0) = 0$ and $\chi(0) < \infty$ 
then $S=\exp(X)$
is conservative. 
\item[$\bullet$] if $S$ is conservative, 
$F(1,0) = R(1,0) = 0$ and $\chi(1) < \infty$, 
then $S=\exp(X)$ is a martingale.
\end{enumerate}
\end{remarks}

Since 
$S$ 
serves as a forward price process under a risk-neutral
measure 
$\PP$
in an arbitrage-free asset pricing model,
it has to be conservative and a martingale
and hence we assume:
\begin{description}
\item[A3] $F(0,0) = R(0,0) = F(1,0) = R(1,0) = 0$ and
$\chi(0)+\chi(1)<\infty$. 
\end{description}
In particular
$F(0,0) = R(0,0) = 0$
in assumption A3
means that the instantaneous killing rates
$c$
and
$\gamma$
in~\eqref{Eq:FR_form} 
are zero and the condition
$F(1,0) = R(1,0) = 0$ is closely related to the functions
$\psi(\cdot,1,0)$ and $\phi(\cdot,1,0)$ being identically
equal to zero (see the generalized Riccati equations 
in~\eqref{Eq:gen_Riccati}),
which implies the martingale property of $S=\exp(X)$.
The following non-degeneracy assumption will guarantee the
stochasticity of volatility of the process
$X$.

\begin{description}
\item[A4] There exists some $u \in \bbR$, such that
$R(u,0) \neq 0$.
\end{description}

\begin{definition}\label{Def:ASVM}
The process $(X,V)$ is a 
\textit{non-degenerate} 
(resp.~\textit{degenerate})
\textit{affine stochastic volatility process} 
if it satisfies assumptions A1~--~A4
(resp.
A1~--~A3 and does not satisfy A4)
and 
$S=\E^X$ 
is the corresponding \textit{affine stochastic volatility model}.
\end{definition}

\begin{remark} 
Assumption A4 excludes the degenerate case  where the distribution 
of $X$ does not depend on the volatility state $V_0$. Indeed, if 
A4 is not satisfied, i.e. 
$R(\cdot,0)\equiv 0$, then~\eqref{Eq:gen_Riccati}
implies that 
$\psi(t,u,0)=0$
and
$\phi(t,u,0)=tF(u,0)$
for all 
$(t,u)\in\Rplus\times\bbC$
where the expectation in assumption~A2
exists.
Hence
if A4 does not hold, then~A2, ~\eqref{Eq:F_form}
and the characterisation theorem for regular affine
processes~\cite[Theorem~2.7]{Schachermayer}
imply that 
$S=\E^X$
is an exponential L\'evy model.
In particular the class of affine stochastic volatility models includes
the Black-Scholes model as a degenerate case.
\end{remark}

The following proposition describes certain properties 
of 
$F$
and
$R$
that will play a crucial 
role in Section~\ref{subsec:Non_degen}.

\begin{proposition}
\label{prop:Properties_F_and_R} 
Let $(X,V)$ be a non-degenerate affine stochastic volatility model and
let the sets
$\cD_F=\left\{(u,w)\in\bbR^2:F(u,w)<\infty\right\}$
and
$\cD_R=\left\{(u,w)\in\bbR^2:R(u,w)<\infty\right\}$
be the effective domains of the functions
$F$
and
$R$
respectively. Then the following holds:
\begin{enumerate}[(A)]
\item $F$ and $R$ are lower semicontinuous convex functions, which are continuously differentiable in 
the interiors 
$\cD_F^\circ$
and 
$\cD_R^\circ$
(in
$\bbR^2$),
and their effective domains 
$\cD_F$
and 
$\cD_R$ are also convex; 
\item $F$ and $R$ are either affine or strictly convex functions 
when restricted to one-dimensional affine subspaces of 
$\bbR^2$. 
\end{enumerate}
\end{proposition}

\begin{proof}
The L\'evy-Khintchine representation for 
$F$
and
$R$
in~\eqref{Eq:FR_form}
implies that they are cumulant generating functions of some
(infinitely divisible) random vectors taking values
in $\bbR^2$. 
H\"older's inequality yields
that 
$F$
and
$R$
are convex. 
The dominated convergence theorem and 
the representation in~\eqref{Eq:FR_form}
implies that
$F$
and
$R$
are analytic in 
$\cD_F^\circ$
and 
$\cD_R^\circ$
respectively. 
Fatou's lemma implies that the functions
$F$
and
$R$
are lower semicontinuous.
Since  
$F$
and
$R$
are cumulant generating functions,
the second derivative of their restriction to an affine subspace
is either identically zero or strictly positive everywhere
(each affine subspace in 
$\bbR^2$ corresponds to a random variable which
takes values in $\bbR$ and may or may not be constant almost surely).
This concludes the proof.
\end{proof}


\subsection{SDE representation of affine stochastic volatility processes}
In order to define an affine stochastic volatility model
one needs to choose admissible parameters 
$(a, \alpha, b, \beta, m, \mu)$
such that the corresponding process
$(X,V)$,
which 
exists by~\cite[Thm.~2.7]{Schachermayer},
satisfies 
assumptions A1 -- A3
(note that 
$c=\gamma=0$ 
by A3 and will henceforth be
ignored). This procedure yields a semigroup,
and hence the law, of the Markov process
$(X,V)$ which is in principle sufficient for
option pricing. However path-wise descriptions
of the pricing models in financial markets are 
widely used as they add to the intuitive understanding
of the properties of the model.
In the rest of this section 
we briefly describe a path-wise construction of the process
$(X,V)$,
given in~\cite{Dawson2006},
and relate it to the most popular affine stochastic
volatility models used in derivatives pricing.

Assume that 
the parameters 
$(a, \alpha, b, \beta, m, \mu)$ 
are admissible 
and suppose in addition that the tails (i.e. the large jumps) of
$m$
and
$\mu$
satisfy:
\begin{eqnarray}
\label{eq:BigJumpCond}
\int_{\|\xi\| > 1}{\|\xi\|\,m(\dd\xi)} < \infty\qquad\text{and}\qquad \int_{\|\xi\| >
1}{\|\xi\|\,\mu(\dd\xi)} <
\infty
\end{eqnarray}
where
$\|\xi\|=\langle \xi,\xi\rangle.$
Let
\begin{equation}
\label{eq:Reparamterisation}
\begin{split}
\widetilde{b}_1 & = b_1 +\int_{D \setminus \{0\}}\frac{\xi_1^3}{1+\xi_1^2}\,m(\dd \xi),\quad
\widetilde{b}_2=b_2,\\
\widetilde{\beta}_i & = \beta_i +\int_{D \setminus \{0\}}\frac{\xi_i^3}{1+\xi_i^2}\,\mu(\dd
\xi)\quad\text{for}\quad i=1,2.
\end{split}
\end{equation}
Note that
the integrals in~\eqref{eq:Reparamterisation} 
are finite by~\eqref{eq:BigJumpCond}
and the parameters 
$(a, \alpha, \widetilde{b}, \widetilde{\beta}, m, \mu)$ 
are admissible with appropriate truncation functions.

Let  
$(\Omega, \cF, (\cF_t)_{t\geq0}, \PP)$ be a filtered probability space equipped with
\begin{enumerate}
\item[$\bullet$] a three-dimensional standard Brownian motion $(B^0, B^1, B^2)$,
\item[$\bullet$] a Poisson random measure $N_0(\dd s,\dd\xi)$ on $\Rplus \times D$ with compensator
$\dd s\,m(\dd\xi)$,
\item[$\bullet$] a Poisson random measure $N_1(\dd s,\dd u,\dd\xi)$ on $\Rplus^2 \times D$ with
compensator $\dd s\,\dd u\,\mu(\dd\xi)$,
\end{enumerate}
where as usual 
$D = \bbR \times \Rplus$
denotes the state-space of the model. 
Let
$$\widetilde{N}_{0}(\dd s,\dd\xi)=N_0(\dd s,\dd\xi)-\dd s\,m(\dd\xi)\quad\text{and}\quad
\widetilde{N}_{1}(\dd s,\dd u,\dd\xi)=N_1(\dd s,\dd u, \dd\xi)-\dd s\,\dd u\,\mu(\dd\xi)$$ 
be the compensated Poisson random measures and let 
$\sigma$ 
be a 
$2 \times 2$-matrix such
that $\sigma \sigma^\top = \alpha$. 
Theorem~6.2 in~\cite{Dawson2006}
implies that the system of SDEs
\begin{eqnarray}
\label{eq:SDE_X}
\dd X_t & = & \left(\widetilde{b}_1 + V_t \widetilde{\beta}_1\right)\dd t + \sqrt{a_{11}} \dd B_t^0 +
\sqrt{V_t}\sigma_{11}\dd B_t^1 + \sqrt{V_t}\sigma_{12}\dd B_t^2 + \\
& & \int_{D \setminus \{0\}}  \xi_1 \widetilde{N}_0(\dd t,\dd\xi) + \int_{D \setminus \{0\}}   \int_0^{V_{t-}}\xi_1 
\widetilde{N}_1(\dd t,\dd u,\dd\xi),\nonumber \\
\label{eq:SDE_V}
\dd V_t & = & \left(\widetilde{b}_2 + V_t \widetilde{\beta}_2\right)\dd t + \sqrt{V_t}\sigma_{21}\dd B_t^1 +
\sqrt{V_t}\sigma_{22}\dd B_t^2 + \\
& &  \int_{D \setminus \{0\}} \xi_2  N_0(\dd t,\dd\xi) + \int_{D \setminus \{0\}}  \int_0^{V_{t-}} \xi_2 \widetilde{N}_1(\dd
t,\dd u,\dd\xi),\nonumber
\end{eqnarray}
with initial condition 
$(X_0,V_0)\in \bbR\times(0,\infty)$,
has a unique strong solution 
$(X,V)$
that is an affine Markov process with admissible parameters
$(a, \alpha, b, \beta, m, \mu)$.
\begin{remarks}
\noindent (i) The change of parameters 
$b$
and
$\beta$
introduced in~\eqref{eq:Reparamterisation}
is inessential. Its function is to establish the
notational compatibility with~\cite{Dawson2006}.

\smallskip
\noindent (ii) The integrals 
in~\eqref{eq:SDE_X}--\eqref{eq:SDE_V}
against
$\widetilde{N}_1$
are taken over a random set 
whose 
$\dd s\,\dd u\,\mu(\dd\xi)$-volume 
is proportional to
$V_{t-}$. 
This, together with the structure of the Poisson random measure
$N_1$,
reinforces the intuition that the jumps 
of the process
$(X,V)$
that
correspond to the integral
term with respect to 
$\widetilde{N}_1$
have random intensity which is proportional to
$V$.
\end{remarks}

\subsection{Examples of affine stochastic volatility models}
\label{subsec:Examples}
We now describe some of the affine stochastic volatility models 
that are of interest in the financial markets and can be obtained
as solutions of the
special cases of SDE~\eqref{eq:SDE_X}--\eqref{eq:SDE_V}.

\subsubsection{Heston model}
The log-price $X$ and the stochastic 
variance process $V$ are given under the
risk-neutral measure by the SDE
\begin{align*}
\dd X_t &= -\frac{V_t}{2}\,\dd t + \sqrt{V_t}\,\dd W^1_t,\\
\dd V_t &= -\lambda(V_t - \theta)\,\dd t + \zeta \sqrt{V_t}\,\dd W^2_t,\,
\end{align*}
where $W^1, W^2$ are Brownian motions with correlation
parameter $\rho\in(-1,1)$, and $\zeta, \lambda, \theta > 0$
(see~\cite{Heston}). 
The affine characteristics of the model are
\begin{subequations}\label{Eq:FR_Heston}
\begin{align}
F(u,w) &= \lambda \theta w,\\
R(u,w) &= \frac{1}{2}(u^2 - u) + \frac{\zeta^2}{2}w^2 - \lambda w
+ uw \rho \zeta.
\end{align}
\end{subequations}
It is easily seen that $\chi$ is given by 
\begin{equation}
\label{eq:Heston_chi}
\chi(u) = \rho \zeta u - \lambda
\end{equation}
and it is trivial to check that A1 -- A4 are satisfied.

\subsubsection{Heston model with state-independent jumps}
\label{sec:Heston_Jumps}
Let 
$J$ 
be a pure-jump L\'evy process
independent of the correlated Brownian motions
$W^1$
and
$W^2$.
The Heston-with-jumps model is defined by the SDEs
\begin{align*}
\dd X_t &= \left(\delta - \frac{V_t}{2}\right)\,\dd t + \sqrt{V_t}\,\dd W^1_t + \dd
J_t,\\
\dd V_t &= -\lambda(V_t - \theta)\,\dd t + \zeta\sqrt{V_t}\,\dd W^2_t,
\end{align*}
where $\zeta, \lambda, \theta > 0$
and
$\delta\in\bbR$.
Assume that 
$J$
is a spectrally negative L\'evy process with characteristic exponent
$\frac{1}{t}\log \EE[\E^{uJ_t}]=\int_{(-\infty,0)}(\E^{u\xi_1}-1-u\xi_1/(1+\xi_1^2))\,\nu(\dd \xi_1)$.
Since 
$J$
only jumps down, this 
captures
the generic situation in the modelling 
of equity markets. 
Assume further that the jumps of 
$J$
are integrable (i.e. $\int_{(-\infty,0)}|\xi_1|\,\nu(\dd \xi_1)<\infty$).
In order to identify the coefficients in~\eqref{eq:SDE_X}--\eqref{eq:SDE_V},
we compensate 
$J$
so that it becomes a martingale and can hence be expressed as an integral
against the compensated Poisson random measure
$\widetilde{N}_0$.
This
implies 
$$\widetilde{b}_1=\delta+\int_{(-\infty,0)}\frac{\xi_1^3}{1+\xi_1^2}\,\nu(\dd\xi_1).$$
It is easily seen that 
$a=0$,
$\alpha_{11}=1$, $\alpha_{12}=\rho\zeta$, $\alpha_{22}=\zeta^2$,
$\mu\equiv0$
and
$m(\dd \xi)=(\nu\otimes\delta_0)(\dd\xi)$,
where
$\delta_0$
is the Dirac delta measure.
Therefore 
$b_1=\delta$,
$b_2=\lambda\theta$,
$\beta_1=-1/2$
and
$\beta_2=-\lambda$.
The martingale condition ($F(1,0)=0$)
implies 
$\delta= -\int_{(-\infty,0)}\left(\E^{\xi_1} - 1-\xi_1/(1+\xi_1^2)\right)\,\nu(\dd \xi_1)$
and the affine form of
the model is given by
\begin{subequations}\label{Eq:HEston_Jumps}
\begin{align}
\label{eq:F_HEston_Jumps}
F(u,w) &= \lambda \theta w + \widetilde{\kappa}(u),\\
\label{eq:R_HEston_Jumps}
R(u,w) &= \frac{1}{2}(u^2 - u) + \frac{\zeta^2}{2}w^2 - \lambda w
+ uw \rho \zeta,
\end{align}
\end{subequations}
where $\widetilde{\kappa}(u)$ is the compensated cumulant generating
function of the jump part, i.e.
\begin{equation}
\label{eq:Kappa_Def}
\widetilde{\kappa}(u) =
\int_{(-\infty,0)}\left(\E^{\xi_1u} - 1 - u \left(\E^{\xi_1} -
1\right)\right)\,\nu(\dd\xi_1).
\end{equation}

\subsubsection{A model of Bates with state-dependent jumps}
\label{subsec:Bates}
We consider the model given by
\begin{align*}
\dd X_t &= -\left(\frac{1}{2}+\delta \right)V_t\,\dd t + \sqrt{V_t}\,\dd W^1_t +
\int_{\bbR\setminus\{0\}}{\xi_1\,\widetilde{N}(V_t,\dd t,\dd \xi_1)},\\
\dd V_t &= -\lambda(V_t - \theta)\,\dd t + \zeta\sqrt{V_t}\,\dd W^2_t,
\end{align*}
where as before $\lambda, \theta, \zeta > 0$,
$\delta\in\bbR$
and the Brownian
motions $W^1$
and
$W^2$
are correlated with correlation $\rho\in(-1,1)$. 
The jump component
is given by $\widetilde{N}(V_t,\dd t,\dd \xi_1) = N(V_t,\dd t,\dd \xi_1) -
n(V_t,\dd t,\dd \xi_1)$, where $N(V_t,\dd t,\dd \xi_1)$ is a Poisson random measure
independent of 
$W^1$
and
$W^2$
with intensity measure $n(V_t,\dd t,\dd \xi_1)$ of the
\emph{state-dependent} form $V_t \nu(\dd \xi_1) \dd t$. Here 
$\nu(\dd \xi_1)$ 
denotes
a L{\'e}vy measure on 
$\bbR\setminus\{0\}$.
A model of this kind has been proposed in~\cite{Bates} 
to explain the time-variation of
jump-risk implicit in observed option prices. 

As in Section~\ref{sec:Heston_Jumps}, 
we assume that 
the support of
$\nu(\dd \xi_1)$ 
is contained in 
$(-\infty,0)$
and that 
the inequality
$\int_{(-\infty,0)}|\xi_1|\,\nu(\dd \xi_1)<\infty$
is satisfied.
We can identify the parameters in~\eqref{eq:SDE_X}--\eqref{eq:SDE_V}
as
$a=0$,
$\alpha_{11}=1$, $\alpha_{12}=\rho\zeta$, $\alpha_{22}=\zeta^2$,
$\widetilde{\beta}_1=-1/2-\delta$,
$\widetilde{\beta}_2=-\lambda$,
$\widetilde{b}_1=0$,
$\widetilde{b}_2=b_2= \lambda\theta$,
$m\equiv0$
and
$\mu(\dd \xi)=(\nu\otimes\delta_0)(\dd\xi)$,
where
$\delta_0$
is the Dirac delta concentrated at
$0$.
Hence we find
$$\beta_1=-\frac{1}{2}-\delta-\int_{(-\infty,0)}\frac{\xi_1^3}{1+\xi_1^2}\,\nu(\dd\xi_1)$$
and 
$\beta_2=-\lambda$,
$b_1=0$.
The functions $F$ and $R$
for the Bates model are
\begin{subequations}\label{Eq:Bates}
\begin{align} 
\label{eq:F_Bates}
F(u,w) &= \lambda \theta w,\\
R(u,w) &= \frac{1}{2}(u^2 - u) + \frac{\zeta^2}{2}w^2 - \lambda w + uw \rho \zeta + \widetilde{\kappa}(u),
\label{eq:R_Bates}
\end{align}
\end{subequations}
where $\widetilde{\kappa}(u)=
\int_{(-\infty,0)}\left(\E^{\xi_1u} - 1 - u \left(\E^{\xi_1} -
1\right)\right)\,\nu(\dd\xi_1)$ 
and the martingale property
($R(1,0)=0$)
was used to determine the value of the parameter
$\delta= \int_{(-\infty,0)}\left(\E^{\xi_1} - 1-\xi_1 \right)\,\nu(\dd \xi_1)$.
It is clear that 
$\chi(u) = \rho \zeta u - \lambda$
and that A1 -- A4 are satisfied.

\subsubsection{The Barndorff-Nielsen-Shephard (BNS) model}
\label{subsec:BNS_Model}
The BNS model was introduced in~\cite{BNS} as a model for 
asset pricing. Under a risk-neutral measure, it can be defined
by the following SDE
\begin{align*}
\dd X_t &= (\delta - \frac{1}{2} V_t)\dd t + \sqrt{V_t}\,\dd W_t +
\rho\,\dd J_{\lambda t},\\
\dd V_t &= -\lambda V_t \,\dd t + \dd J_{\lambda t},
\end{align*}
where $\lambda > 0$, $\rho < 0$  and $(J_t)_{t \geq 0}$ is a
L\'evy subordinator with the L\'evy measure
$\nu$, i.e. a pure jump L\'evy process that
increases a.s. 
The cumulant 
generating function 
$\kappa(u)$
of $(J_t)_{t \geq 0}$
takes the form
\begin{eqnarray}
\label{eq:Kappa_BNS}
\kappa(u)=\int_{(0,\infty)}(\E^{u\xi_2}-1)\,\nu(\dd\xi_2).
\end{eqnarray}
To conform with~\eqref{eq:BigJumpCond}
we further assume that 
$\int_{(0,\infty)}\xi_2\,\nu(\dd \xi_2)<\infty$.
The drift $\delta$ will be determined by the martingale
condition for $S$. The time-scaling $J_{\lambda t}$ is
introduced in~\cite{BNS} to make the invariant distribution
of the variance process independent of $\lambda$. The distinctive
features of the BNS model are that the variance process has no
diffusion component, i.e. moves purely by jumps, and the
negative correlation between variance and price movements is
achieved by simultaneous jumps in $V$ and $X$. 

It follows from~\eqref{eq:SDE_X}--\eqref{eq:SDE_V}
and the SDE above that
$a=\alpha_{12}=\alpha_{22}=0$,
$\alpha_{11}=1$,
$\mu\equiv0$
and
$$
m(\dd\xi) = I_{\{\xi_1=\rho\xi_2\}} \lambda\nu(\dd\xi_2),
$$
where
$I_{\{\xi_1=\rho\xi_2\}} $
denotes the indicator function of 
the half-line $\xi_1=\rho\xi_2$
in 
$D\setminus \{0\}$.
Therefore it follows that 
$\widetilde{\beta}_1=\beta_1=-1/2$,
$\widetilde{\beta}_2=\beta_2=-\lambda$,
$\widetilde{b}_2=b_2=0$,
$\widetilde{b}_1=\delta+\lambda\int_{(0,\infty)}\rho\xi_2\,\nu(\dd\xi_2)$
and
$$
b_1=\delta+\lambda\int_{(0,\infty)}\frac{\rho\xi_2}{1+(\rho\xi_2)^2}\,\nu(\dd\xi_2).
$$
The definition of $F$
in~\eqref{Eq:FR_form}
and the martingale condition
$F(1,0)=0$
imply that we need to define
$\delta=-\lambda\kappa(\rho)$,
where
$\kappa$
is the cumulant generating function of 
$J$
given in~\eqref{eq:Kappa_BNS}.
The BNS model
is an affine stochastic volatility model with  $F$ and $R$ 
given by
\begin{subequations}\label{Eq:BNS}
\begin{align}
\label{eq:F_BNS}
F(u,w) &= \lambda \kappa(w + \rho u) - u \lambda \kappa(\rho),\\
R(u,w) &= \frac{1}{2}(u^2 - u) - \lambda w.
\label{eq:R_BNS}
\end{align}
\end{subequations}
We have
$\chi(u) = - \lambda$
and the assumptions 
A1 -- A4 are clearly satisfied.

\section{Large deviation principle and the G\"artner-Ellis theorem}
\label{sec:LargeDeviations}

In this section we give a brief review of the key concepts of 
large deviations for a family of (possibly dependent) random variables
$(Z_t)_{t\geq1}$
and state a version of the G\"artner-Ellis theorem
(see Theorem~\ref{thm:GartnerEllis}) 
that will be used to obtain the asymptotic behaviour of the
option prices and implied volatilities. 
A general reference for all the concepts in this
section is~\cite[Section 2.3]{Dembo2010}.

Let 
$Z_t$
take values in 
$\bbR$
and 
recall that
$I:\bbR\to(-\infty,\infty]$
is \textit{lower semicontinuous}
if
$\{x:I\left(x\right)\leq\alpha\}$
is closed in
$\bbR$ for any
$\alpha\in\bbR$
(intuitively
for any
$x_0\in\bbR$
the values of
$I$
near
$x_0$
are either close to
$I(x_0)$
or greater than
$I(x_0)$).
A nonnegative 
lower semicontinuous
function 
$I$
is called a \textit{rate function}.
If in addition 
$\{x:I\left(x\right)\leq\alpha\}$
is compact for any
$\alpha\in\bbR$,
then 
$I$
is a \textit{good rate function}.

\begin{definition}
\label{eq:DefLDP}
The family
$(Z_t)_{t\geq1}$
satisfies the \textit{large deviation principle (LDP)}
with the \textit{rate function}
$I$
if for every Borel set
$B\subset\bbR$
we have
\begin{equation*}
-\inf\{I(x):x\in B^\circ\}\leq\liminf_{t\to\infty}\frac{1}{t}\log \PP\left[Z_t\in B\right]
\leq
\limsup_{t\to\infty}\frac{1}{t}\log \PP\left[Z_t\in
B\right]\leq-\inf\left\{I(x):x\in \overline B\right\},
\end{equation*}
with the convention
$\inf\emptyset =\infty$
(the interior 
$B^\circ$
and closure 
$\overline B$
are relative to the topology of 
$\bbR$).
\end{definition}
An important consequence of 
Definition~\ref{eq:DefLDP}
is that if
$(Z_t)_{t\geq1}$ 
satisfies LDP \textit{and} 
$I$
is continuous on 
$\overline B$,
then 
$\lim_{t\to \infty}t^{-1}\log \PP\left[Z_t\in B\right]=
-\inf\{I(x):x\in B\}$.
\smallskip


The G\"artner-Ellis theorem (Theorem~\ref{thm:GartnerEllis}) gives sufficient 
conditions for 
$(Z_t)_{t\geq1}$
to satisfy
the LDP and in that case describes the rate function.
Let 
$\Lambda_t^Z(u):=\log \EE\left[\E^{uZ_t}\right]$
be a
cumulant generating function.
Assume that for every
$u\in\bbR$
\begin{eqnarray}
\label{eq:0_in_Domain_Assumption}
\Lambda(u) & := & \lim_{t\to \infty}\frac{1}{t}\Lambda_t^Z(tu)
\quad\text{exists in } 
[-\infty,\infty]
\qquad\text{and}\qquad 0  \in  \cD_\Lambda^\circ,
\end{eqnarray}
where
$\mathcal D_\Lambda:=\{u\in\mathbb R\>:\>\Lambda(u)<\infty\}$
is the \textit{effective domain}
of 
$\Lambda$ 
and
$\cD_\Lambda^\circ$
is its interior in
$\bbR$.
Since 
$\Lambda_t^Z$
is convex (by the H\"older inequality)
for every 
$t$,
the limit 
$\Lambda$
is also convex
by~\cite[Theorem~10.8]{Rockafellar}
and the set
$\mathcal D_\Lambda$
is an interval.
Since
$\Lambda(0)=0$,
convexity of
$\Lambda$
and
$0\in\cD_\Lambda^\circ$
imply
$\Lambda(u)>-\infty$
for all
$u\in\bbR$.
Furthermore the convexity implies that
$\Lambda$
is continuous on
$\mathcal D_\Lambda^\circ$.
{
The statement in~\eqref{eq:0_in_Domain_Assumption} is an important assumption of 
G\"artner-Ellis theorem
(Theorem~\ref{thm:GartnerEllis} below), which
in particular implies 
$\mathcal D_\Lambda^\circ \neq \emptyset$.
However the converse does not hold in general,
i.e. 
if
$0$
is a boundary point of a domain 
$\mathcal D_\Lambda$
with non-empty interior, LDP may
still hold true. 

A further property of the function 
$\Lambda:\bbR\to(-\infty,\infty]$,
which arises as an assumption in Theorem~\ref{thm:GartnerEllis},
is essential smoothness.
\begin{definition}
\label{def:Essential_Smoothness}
A convex function 
$\Lambda:\bbR\to(-\infty,\infty]$
is \textit{essentially smooth}
if 
\begin{itemize}
\item[(a)]
$\mathcal D_\Lambda^\circ$
is non-empty;
\item[(b)]
$\Lambda$ is
differentiable in 
$\mathcal D^\circ_\Lambda$;
\item[(c)] 
$\Lambda$ is \textit{steep},
in other words it satisfies 
$\lim_{n\to\infty}|\Lambda'(u_n)|=\infty$
for every sequence 
$(u_n)_{n\in\mathbb N}$
in
$\mathcal D^\circ_\Lambda$
that converges to a boundary point of
$\mathcal D^\circ_\Lambda$.
\end{itemize}
\end{definition}

The \textit{Fenchel-Legendre transform}
(or \textit{convex dual})
$\Lambda^*$
of 
$\Lambda$
is
defined by the formula
\begin{eqnarray}
\label{eq:DefFenchelLegendreTransf}
\Lambda^*(x) &:=& \sup\{ux-\Lambda(u)\>:\>u\in\bbR\}
\quad\text{for}\quad
x\in\bbR
\end{eqnarray}
with an effective domain
$\mathcal D_{\Lambda^*}:=\{x\in\mathbb R\>:\>\Lambda^*(x)<\infty\}$.
The following properties are immediate from the definition:
\begin{itemize}
\item[(i)] $0\leq \Lambda^*(x)\leq\infty$ for all $x\in\bbR$, since $\Lambda(0)=0$;
\item[(ii)]  $\Lambda^*(x) = \sup\{ux-\Lambda(u):u\in\mathcal D_\Lambda \}$
for all $x\in\bbR$ 
and  hence
$\Lambda^*$
is convex in the interval $\mathcal D_{\Lambda^*}$
and continuous in the interior $\mathcal D_{\Lambda^*}^\circ$; 
\item[(iii)]  $\Lambda^*$ is  lower semicontinuous on 
$\bbR$ as it is a supremum of continuous (in fact linear) functions.
Hence the level sets
$\{x:\Lambda^*(x)\leq\alpha\}$
are closed. 
\end{itemize}
In general 
$\mathcal D_{\Lambda^*}$
can be strictly contained
in $\bbR$
and
$\Lambda^*$
can be 
discontinuous 
at the boundary of 
$\mathcal D_{\Lambda^*}$
(see~\cite[Section 2.3]{Dembo2010}
for elementary examples of such rate functions).
Assumption~\eqref{eq:0_in_Domain_Assumption}
implies that 
for any
$\delta>0$,
such that
$(-\delta,\delta)\subset\mathcal D_\Lambda^\circ$,
and 
$c=\sup\{\Lambda(u):u\in[-\delta,\delta]\}$
we have
\begin{eqnarray}
\label{eq:Compac_LEvel_Sets}
\Lambda^*(x) \quad \geq \quad  \sup\{ux-\Lambda(u):u\in[-\delta,\delta]\} \quad \geq \quad 
\delta |x| -c.
\end{eqnarray}
Hence the set
$\{x:\Lambda^*(x)\leq\alpha\}$
is compact for any
$\alpha\in\bbR$
and therefore
$\Lambda^*$
is a good rate function.

\begin{remarks} (A) If 
$\Lambda$
is strictly convex, differentiable 
on 
$\mathcal D_\Lambda^\circ$
and steep,
which is the case in the
applications in this paper, 
then 
$\mathcal D_{\Lambda^*}=\bbR$
and for each
$x\in\bbR$
the equation 
$\Lambda'(u)=x$
has a unique solution 
$u_x$
in 
$\mathcal D_\Lambda^\circ$.
Furthermore the formula 
\begin{eqnarray}
\label{eq:FromulaForTrnsform}
\Lambda^*(x) & = & x u_x - \Lambda(u_x)
\end{eqnarray}
holds. This reduces the computation of 
$\Lambda^*(x)$
to finding the unique root of the 
equation
$\Lambda'(u)=x$,
where the strictly increasing function 
$\Lambda'$
is in most applications known in closed form.\\
\noindent (B) If 
$(Z_t)_{t\geq1}$
satisfies~\eqref{eq:0_in_Domain_Assumption}
and
the function 
$\Lambda$
satisfies the assumptions of Remark~(A)
and is twice differentiable with 
$\Lambda''(u)>0$
for all 
$u\in\mathcal D_\Lambda^\circ$,
then~\eqref{eq:FromulaForTrnsform}
implies that the Fenchel-Legendre transform
$\Lambda^*$
is differentiable with the derivative
\begin{eqnarray}
\label{eq:DerivativeOfRateFunction}
\left(\Lambda^*\right)'(x) & = & 
\left(\Lambda'\right)^{-1}(x) \quad\text{ for all }\quad x\in\bbR.
\end{eqnarray}
In particular~\eqref{eq:DerivativeOfRateFunction}
implies that 
$\Lambda^*$
is strictly convex on
$\bbR$
and that 
its global minimum is attained at 
the unique  point
$x^*$
given by
\begin{eqnarray*}
x^* & =  & \Lambda'(0). 
\end{eqnarray*}
\end{remarks}

We state a simple version of the G\"artner-Ellis theorem
(for the proof see~\cite[Section 2.3]{Dembo2010}).

\begin{theorem}
\label{thm:GartnerEllis}
Let 
$(Z_t)_{t\geq1}$
be a family of random variables 
that satisfies assumption~\eqref{eq:0_in_Domain_Assumption}
with the limiting cumulant generating function
$\Lambda:\bbR\to(-\infty,\infty]$.
If
$\Lambda$
is essentially smooth and lower semicontinuous, then the LDP
holds for
$(Z_t)_{t\geq1}$
with the good rate function
$\Lambda^*$.
\end{theorem}

\section{Limiting cumulant generating function in affine stochastic volatility models}
\label{sec:Cumulant}

\subsection{Non-degenerate affine stochastic volatility processes}
\label{subsec:Non_degen}
Let
$(X,V)$
be a non-degenerate affine stochastic volatility process 
(see Definition~\ref{Def:ASVM}).
The goal of the present section is 
to describe the limiting cumulant generating function
$h$
of the family of variables
$(X_t/t)_{t\geq1}$,
defined by
\begin{eqnarray}
\label{eq:DefOfh}
h(u) & := & \lim_{t\to\infty}\frac{1}{t} \log \EE\left[\E^{uX_t}\right]
\end{eqnarray}
for every 
$u\in\bbR$
where the limit in~\eqref{eq:DefOfh} exists as an extended real number.
The function 
$h$
will determine the limiting implied volatility smile of the model
$S=\E^X$.
To ensure that 
$h$
is finite on an interval that contains
$[0,1]$,
which is key for establishing the LDP,
a further assumption will be required: 
\begin{description}
\item[A5] $\chi(0)<0$ and $\chi(1)<0$, where
$\chi$ 
is given in Definition~\ref{def:chi}.
\end{description}
This assumption will also imply that 
$h$
can be uniquely extended to a cumulant generating 
function of an infinitely divisible random variable.

In order to apply the G\"artner-Ellis theorem in
our setting, we need to answer the following 
three questions:
is 
$h$
well-defined 
as an extended real number
by~\eqref{eq:DefOfh}
for every 
$u\in\bbR$,
does the effective domain 
$\cD_h$
contain 
$[0,1]$
in its interior 
\textit{and} is 
$h$
essentially smooth?
Answers to these questions play 
a crucial role in establishing the large deviation principle,
via Theorem~\ref{thm:GartnerEllis},
for affine stochastic volatility models.
Theorem~\ref{thm:LCGFh} and Corollary~\ref{cor:LCGFh},
proved in this section,
provide easy to check sufficient conditions for the affirmative answers to hold.



It is shown in ~\cite{K2008a}
that the function 
$h$
can be obtained from the functions $F$ and $R$ without 
the explicit knowledge 
of $\phi$ and $\psi$
(see Section~\ref{sec:Affine} for definition of
$\psi,\phi$).
Lemma~\ref{Lem:existence} and Theorem~\ref{Thm:wm_convergence}, 
taken from~\cite[Lemma~3.2 and Theorem~3.4]{K2008a}, 
describe certain properties of 
the limiting cumulant generating function 
$h$,
which are needed in Section~\ref{sec:OptionPrices}
but are insufficient to guarantee the essential smoothness of
$h$.
The main contribution of the present section is 
Theorem~\ref{thm:LCGFh},
which identifies sufficient conditions for the process
$(X,V)$
that imply essential smoothness of the function 
$h$.
The conditions in
Theorem~\ref{thm:LCGFh}
are easy to apply to the models of Section~\ref{subsec:Examples},
which will allow us to find their limiting implied volatility
smiles.

\begin{lemma}\label{Lem:existence}
Let $(X,V)$ be a non-degenerate affine stochastic volatility
process that satisfies assumption~A5.
Then there exist a maximal interval 
$\cI$
and a 
unique convex function
$w:\cI \to\bbR$
such that
$w \in C(\cI)\cap C^1(\cI^\circ)$ 
and
\begin{equation*}
R(u,w(u)) = 0 \qquad \text{for all} \quad u \in \cI,
\end{equation*}
where
$R$
is given in~\eqref{eq:F_and_R} (see also~\eqref{Eq:FR_form}).
Furthermore we have
\begin{enumerate}[(a)]
\item $[0,1] \subseteq \cI$ and $\partial_2R(u,w(u)) < 0$ for all $u \in \cI^\circ$;
\item $w(0) = w(1) = 0$ and $w(u) < 0$ for all $u \in (0,1)$; 
\item $w(u) > 0$ for all $u \in \cI \setminus [0,1]$. 
\end{enumerate} 
\end{lemma}

\begin{remarks}
\noindent (i) The proof of Lemma~\ref{Lem:existence} 
in~\cite{K2008a}
is based on the analysis of the qualitative properties of the
generalized Riccati equations in~\eqref{Eq:gen_Riccati}.

\smallskip
\noindent (ii) The function 
$u\mapsto w(u)$
from Lemma~\ref{Lem:existence}
can be extended naturally to a
lower semicontinuous function
$w:\bbR \to(-\infty,\infty]$
by 
$w\left(\bbR\setminus\cI\right)=\infty$.
Then the extension, again denoted by
$w(u)$,
has the following properties:
\begin{itemize}
\item $w$ is convex with 
effective domain 
$\cD_w=\cI$ and 
$\cI=\{u\in\bbR:w(u)<\infty\}$;
\item the maximality of 
$\cI$
implies that for $u \in \bbR\setminus\cI$
there exists no $w^*\in\bbR$ 
such that 
$R(u,w^*) = 0$.  
\end{itemize}
\end{remarks}

The next theorem, proved 
in~\cite[Theorem~3.4]{K2008a},
describes further properties of the function
$u\mapsto w(u)$
and specifies its relationship to the limiting cumulant generating
function
$h$
defined in~\eqref{eq:DefOfh}.

\begin{theorem}\label{Thm:wm_convergence}
Let $(X,V)$ be a non-degenerate affine stochastic volatility
process that satisfies assumption~A5
and let 
$w(u)$
be given by Lemma~\ref{Lem:existence}. Then 
the 
function
$h(u)$
defined in~\eqref{eq:DefOfh}
satisfies
\begin{eqnarray}
\label{eq:ExpressionFor_h}
h(u) = F(u,w(u)) \quad\text{for any}\quad u\in\cJ := \{s \in \cI: F(s,w(s)) < \infty\},
\end{eqnarray}
where
$F$
is defined in~\eqref{eq:F_and_R} (see also~\eqref{Eq:FR_form}).
Furthermore the inclusions hold:
$[0,1] \subseteq \cJ \subseteq \cI$.
The functions
$w(u)$ 
and 
$h(u)$ 
can be extended uniquely to 
cumulant generating functions of infinitely divisible random
variables and
\begin{subequations}\label{Eq:limit_psiphi}
\begin{align}
\lim_{t \to \infty} \psi(t,u,0) &= w(u) \quad \text{for all} \quad u \in \cI;\label{Eq:limit_psiphi_sub1}\\
\lim_{t \to \infty} \frac{1}{t}\phi(t,u,0) &= h(u) \quad \text{for
all} \quad u \in \cJ.\label{Eq:limit_psiphi_sub2}
\end{align}
\end{subequations}
\end{theorem}

\begin{remark}
Since 
$w(u)$ and $h(u)$ 
can be extended to
cumulant generating functions 
of some (infinitely divisible) random variables 
it follows that:
\begin{itemize}
\item $w$ (resp. $h$) is continuously differentiable in the interior of 
$\cI$
(resp.
$\cJ$);
\item either 
$h''(u)>0$
for all
$u\in\cJ^\circ$
or 
$h(u)=0$
for all
$u\in\bbR$
(this follows from~\eqref{eq:ExpressionFor_h}, (b) in Lemma~\ref{Lem:existence} and assumption A3).
\end{itemize}
\end{remark}

We say that 
$R$
\emph{explodes at the boundary} 
if 
$\lim_{n\to\infty}R(u_n,w_n)=\infty$
for any sequence 
$\left((u_n,w_n)\right)_{n\in\bbN}$
in the interior 
$\cD_R^\circ$
that tends to a point in the boundary of
$\cD_R$
(both the boundary and the interior of 
$\cD_R$
are relative to the topology of
$\bbR^2$)
or equivalently
$\cD_R$
is open.
By Proposition~\ref{prop:Properties_F_and_R}~(A),
the gradient
$\nabla F=(\partial_1 F,\partial_2 F)$
is continuous on 
$\cD_F^\circ$.
Analogously to the one-dimensional case (see (c) in Definition~\ref{def:Essential_Smoothness}),
we say that 
$F$
is \textit{steep} if 
$\lim_{n\to\infty}\|\nabla F(u_n,w_n)\|=\infty$
for any sequence 
$\left((u_n,w_n)\right)_{n\in\bbN}$
in the interior 
$\cD_F^\circ$
that tends to a point in the boundary of
$\cD_F$.
It is clear that if 
$F$
explodes at the boundary, it is also steep 
but the converse may not be true.

Before we state and prove the main results of this section
(Theorem~\ref{thm:LCGFh} 
and 
Corollary~\ref{cor:LCGFh}),
we establish 
Lemma~\ref{lem:h_not_zero},
which states 
that in an affine stochastic volatility model, the limiting cumulant generating function  
$h$
cannot be identically equal to zero. This property will play an important
role in understanding the limiting behaviour of the implied volatility
smile  (see e.g. Theorem~\ref{thm:OptionPrices}).

\begin{lemma}
\label{lem:h_not_zero}
Let $(X,V)$ be a non-degenerate affine stochastic volatility
process that satisfies assumption~A5
and let
$h$
be given by~\eqref{eq:DefOfh}.
Assume further that 
the interior $\cD_F^\circ$
of the effective domain of 
$F$
contains the set
$\{(0,0),(1,0)\}$
and that
$F(u,w)\neq0$
for some
$(u,v)\in\cD_F$.
Then
$h(u)>0$
for all
$u\in\cJ\setminus[0,1]$
and
$h(u)<0$
for all
$u\in(0,1)$.
Furthermore we have
$h''(u)>0$
for all
$u\in\cJ^\circ$.
\end{lemma}

\begin{proof}
Note that since 
$h$
can be extended to 
a cumulant generating function of a random variable
by Theorem~\ref{Thm:wm_convergence},
it is smooth in
$\cJ^\circ$.
Since 
$h$
is either identically equal to zero or strictly convex 
on 
$\cJ$
by the remark following
Theorem~\ref{Thm:wm_convergence}, the statement 
$h''(u)>0$
for all
$u\in\cJ^\circ$
follows if we prove that
$h(u)<0$
for some
$u\in(0,1)$.

The function 
$u\mapsto F(u,0)$
is convex by (B) of Proposition~\ref{prop:Properties_F_and_R}.
Furthermore it is either (I) strictly convex or (II) identically equal to zero
(by A3). We analyse both cases.

(I) Strict convexity and A3 imply that for
$u\in(0,1)$ we have
$F(u,0)<0$.
The same argument implies that 
for 
$u\in\bbR\setminus[0,1]$,
such that 
$(u,0)\in\cD_F^\circ$,
the inequality 
$F(u,0)>0$
holds.
The L\'evy-Khintchine representation of 
$F$
in~\eqref{Eq:FR_form}
implies that 
\begin{eqnarray}
\label{eq:Derivative_In_w}
\partial_2 F(u,w) & = & b_2 +\int_{D\setminus\{0\}}\xi_2\,\E^{u\xi_1+w\xi_2}\,m(\dd\xi)
\end{eqnarray}
for any point in the interior of the effective domain 
$\cD_F$.
It is clear from~\eqref{eq:Derivative_In_w}
that 
$\partial_2 F\geq0$
on 
$\cD_F^\circ$.
Lemma~\ref{Lem:existence}
implies that for
$u\in(0,1)$
we have 
$w(u)<0$.
Identity~\eqref{eq:ExpressionFor_h} in Theorem~\ref{Thm:wm_convergence}
yields
$$
h(u) = F(u,w(u)) = F(u,0) - \int_{w(u)}^0\partial_2 F(u,z)\,\dd z\leq F(u,0)<0.
$$
The last inequality follows from the strict convexity of 
$u\mapsto F(u,0)$.
If
$u\in\cJ^\circ\setminus[0,1]$,
then
$(u,0)\in\cD_F^\circ$
and an analogous argument implies
that 
$h(u)>0$.
The inequality at the boundary points 
of the interval
$\cJ$
follows from the
convexity of 
$h$.

(II) Assume now that 
$u\mapsto F(u,0)$
is identically equal to zero.
For any 
$(u,0)$
in the interior of the effective domain of
$F$,
the L\'evy-Khintchine representation of 
$F$
in~\eqref{Eq:FR_form}
yields
$$
\partial_1^2 F(u,0)  =  a_{11}
+\int_{D\setminus\{0\}}\xi_1^2\,\E^{u\xi_1}\,m(\dd\xi)=0.
$$
This implies 
$a_{11}=0$
and
$m(\dd \xi) = (\delta_0\otimes\nu)(\dd \xi)$,
where
$\nu(\dd\xi_2)$
is a L\'evy measure 
on
$(0,\infty)$
with integrable small jumps
and 
$\delta_0$
is the Dirac delta.
The condition
$F(1,0)=0$
in A3
and 
the representation of 
$F$
in~\eqref{Eq:FR_form}
yield
$b_1=0$.
Hence we have
$$
F(u,w)  =  b_2w +\int_{\Rplus\setminus\{0\}}\left(\E^{w\xi_2}-1\right)\,\nu(\dd\xi_2).
$$ 
Since
by assumption
there exists
$(u,v)\in\cD_F$
such that
$F(u,w)\neq0$, 
either
$b_2>0$
or
$\nu\neq0$ holds.
Therefore
identity~\eqref{eq:ExpressionFor_h}
in Theorem~\ref{Thm:wm_convergence}, 
Lemma~\ref{Lem:existence} 
and this representation of $F$ conclude the proof.
\end{proof}

\begin{remark}
The assumption 
$\{(0,0),(1,0)\} \subset \cD_F^\circ$
in Lemma~\ref{lem:h_not_zero}
ensures that the interiors of the effective domains of 
$F$
and
$h$
are non-empty. It may not be necessary for 
Lemma~\ref{lem:h_not_zero} 
to hold. However, the assumption is crucial in Theorem~\ref{thm:LCGFh}
and hence does not restrict the applicability of Lemma~\ref{lem:h_not_zero}
in our setting.
\end{remark}



\begin{theorem}
\label{thm:LCGFh}
Let $(X,V)$ be a non-degenerate affine stochastic volatility
process that satisfies assumption~A5
and suppose that the function 
$w\mapsto F(0,w)$, 
where 
$F$
is defined in~\eqref{eq:F_and_R},
is not identically equal to zero.
If $R$ explodes at the boundary
(i.e.
$\cD_R$
is open), 
$F$ is steep and $\{(0,0),(1,0)\} \subset \cD_F^\circ$, 
then the 
function $h(u)$ is well-defined by~\eqref{eq:DefOfh}
as an extended real number for every 
$u\in\bbR$
and its effective domain 
is given by
$\cD_h=\cJ$
(see~\eqref{eq:ExpressionFor_h} for the definition
of interval
$\cJ$).
Furthermore 
$h$
is essentially smooth and the set $\{0,1\}$ 
is contained in the interior 
$\cD_h^\circ$ 
(relative to $\bbR$) 
of 
$\cD_h$.
\end{theorem}

\begin{corollary}
\label{cor:LCGFh}
Let $(X,V)$ be a non-degenerate affine stochastic volatility
process that satisfies assumption~A5
and assume that 
$w\mapsto F(0,w)$
is not identically equal to zero.
If either of the following conditions 
holds
\begin{enumerate}[(i)]
\item $\mu$ has exponential moments of all orders, $F$ is steep, and $\cD_F^\circ$
contains $(0,0)$ and $(1,0)$, \label{Cond3}
\item  $(X,V)$ is a diffusion, \label{Cond4}
\end{enumerate}
then the function 
$h$ 
is well-defined by~\eqref{eq:DefOfh}
for every 
$u\in\bbR$
with effective domain 
$\cD_h=\cJ$.
Moreover
$h$
is essentially smooth 
and
$\{0,1\} \subset \cD_h^\circ$.  
\end{corollary}

\noindent \textit{Proof of Corollary~\ref{cor:LCGFh}.}  
Note that either of
the
conditions~\eqref{Cond3}
or~\eqref{Cond4} implies 
that $\cD_R = \bbR^2$ and hence $\cD_R$ is open. 
Therefore~\eqref{Cond3} and the assumptions of Corollary~\ref{cor:LCGFh}
imply the assumptions of Theorem~\ref{thm:LCGFh}.
If~\eqref{Cond4} holds, 
then $(X,V)$ is a diffusion and
\[F(u,w) = a_{11}\frac{u^2}{2} + b_1 u + b_2 w\qquad \text{with}\quad 
a_{11},b_2 \ge
0\quad\text{and}\quad b_1 \in \bbR. \]
Clearly $\cD_F^\circ = \bbR^2$  
contains the set 
$\{(0,0),(1,0)\}$
and
$F$ is steep if $b_2$ is non-zero. 
In the case 
$b_2 = 0$,  
the map
$w \mapsto F(0,w)$
is identically equal to zero, 
which contradicts the assumption in 
Corollary~\ref{cor:LCGFh}.
Thus Corollary~\ref{cor:LCGFh} follows from Theorem~\ref{thm:LCGFh}.~\exit

\medskip

\noindent \textit{Proof of Theorem~\ref{thm:LCGFh}.}  
The proof of this theorem is in two steps. In step (I)
we show that 
$\{0,1\}\subset\cJ^\circ$
and that,
if we extend 
$h|_{\cJ}$
by
$+\infty$
to 
$\bbR\setminus\cJ$,
we obtain an essentially smooth convex function.
In step (II) of the proof we show that 
the limit in definition~\eqref{eq:DefOfh}
exists for any
$u\in\bbR$
as an extended real number and 
that
definition of 
$h$
in~\eqref{eq:DefOfh}
agrees 
for every 
$u\in\bbR$
with the extension of 
$h|_{\cJ}$
from the first part of the proof.

\noindent \textit{Step (I)}.
Throughout this step we abuse notation by using
$h$
to denote the extension of 
$h|_{\cJ}$
to 
$\bbR$
described above.
Theorem~\ref{Thm:wm_convergence}
and the remark following it imply that 
$h$
is essentially smooth (see Definition~\ref{def:Essential_Smoothness})
if it is steep. We will prove the steepness of 
$h$ 
at the right endpoint 
$u_+ = \sup\{u:u\in \cJ\}$
of the interval 
$\cJ$
and show that
$1\in\cJ^\circ$. 
The left endpoint 
$u_- = \inf\{u:u\in \cJ\}$ 
and the fact
$0\in\cJ^\circ$
can be treated by a 
completely symmetrical argument. 

Let 
$(u_n)_{n\in\bbN}$ 
be a sequence in 
$\cJ^\circ$ 
converging to 
$u_+$. 
We use the shorthand notation 
$w_n = w(u_n)$ 
and 
$w_+ = \lim_{n \to \infty} w_n$,
where
$u\mapsto w(u)$
is the function given in Lemma~\ref{Lem:existence}
(note that the limit 
$w_+$
exists but may be infinite since 
$u\mapsto w(u)$
is a cumulant generating function of a random variable
and $\cJ\subset\cI=\cD_w$).
Since $u\mapsto w(u)$ is
convex on 
$\cD_w$,
the value
$w_+$ 
is independent of the choice of sequence
$(u_n)_{n\in\bbN}$.

\noindent \textbf{Claim 1.} The inequalities $u_+ > 1$ and $w_+ > 0$ hold. \\
Indeed, since $R(1,0) = 0$ 
by assumption A3, we get that $(1,0) \in \cD_R=\cD_R^\circ$. 
Assume now that $u_+ = 1$. Then 
by Lemma~\ref{Lem:existence} we have
$w_+ = 0$ 
and 
$(u_+,w_+) \in \cD_R^\circ$. 
Since
$R$ 
is continuously differentiable in 
$\cD_R^\circ$ and $\partial_2 R(1,0)=\chi(1) < 0$ 
by assumption A5, 
the implicit function theorem 
and Lemma~\ref{Lem:existence}
imply that 
$u_+$
is in the set
$\cD_w^\circ=\cI^\circ$.
Since
$(1,0)\in\cD_F^\circ$,
there exists
$u\in\cI^\circ$
such that
$u>u_+$
and
$(u,0)\in\cD_F^\circ$.
Identity~\eqref{eq:ExpressionFor_h}
in Theorem~\ref{Thm:wm_convergence} 
therefore
implies that
$h(u)<\infty$,
which contradicts the definition of
$u_+$.
Therefore $u_+ > 1$. 
Lemma~\ref{Lem:existence}
implies that the sequence 
$(w_n)_{n\in\bbN}$
is eventually (certainly when
$u_n>1$) non-decreasing 
and strictly positive. 
This yields that $w_+ > 0$
and the claim follows.

Discarding finitely many elements we may assume that $u_n > 1$ and $w_n > 0$ for all $n$. 
If $u_+$ is infinite, it is not in the boundary of
$\cJ$
and the steepness of $h$ follows.
If 
$u_+$
is finite but $w_+$ is infinite, 
identity~\eqref{eq:ExpressionFor_h}
and the assumption
that
$w\mapsto F(0,w)$
is non-zero
imply
$\lim_{n\to\infty}h(u_n)=\infty$.
The steepness of $h$ follows from the convexity of
$h$.
Therefore
in the rest of the proof we can assume 
\begin{equation}
\label{eq:Important_Reduction}
u_+\in(1,\infty)\qquad
\text{and}\qquad
w_+\in(0,\infty)
\end{equation}
without loss of generality.

\noindent \textbf{Claim 2.} 
The following statements hold true:
\begin{enumerate}[(a)]
\item \label{Item:a} if $u \in \cI^\circ$, where
$\cI$
is defined in Lemma~\ref{Lem:existence},
then $(u,w(u)) \in \cD_R^\circ$ and
\begin{equation}\label{Eq:R_derivatives}
0 = \partial_1 R(u,w(u)) + \partial_2 R(u,w(u)) w'(u);
\end{equation}
\item \label{Item:b} if $u \in \cJ^\circ\cap(1,\infty)$, where
$\cJ$
is defined in Theorem~\ref{Thm:wm_convergence},
then $(u,w(u)) \in \cD_F^\circ$ and
\begin{equation}\label{Eq:F_derivatives}
h'(u) = \partial_1 F(u,w(u)) + \partial_2 F(u,w(u)) w'(u).
\end{equation}
\end{enumerate}

The statement in~\eqref{Item:a}
follows from Lemma~\ref{Lem:existence},
assumption 
$\cD_R= \cD_R^\circ$
and the chain rule.
To prove the first statement in~\eqref{Item:b},
note that 
$u\in\cJ^\circ\cap(1,\infty)\subset \cI^\circ$
and hence 
$\cI^\circ\setminus[0,1]\neq \emptyset$.
Lemma~\ref{Lem:existence}
therefore implies that the function
$w:\cJ\to\bbR$
is strictly convex
with 
$w(0)=w(1)=0$
and therefore strictly increasing
on $\cJ^\circ\cap(1,\infty)$.
Pick
$u'\in\cJ^\circ\cap(1,\infty)$
such that
$u'>u$
and note that 
$(u',w(u'))\in\cD_F$
(by the definition of 
$\cJ$)
and
$(u',0)\in\cD_F$
(by representation~\eqref{Eq:FR_form} 
and
$0\leq w(u')$).
Assumption 
$(1,0)\in\cD_F$
in the theorem and the fact
$w(u)<w(u')$
imply
that the point 
$(u,w(u))$
lies in the interior of
the triangle with vertices
$(u',w(u')), (1,0), (u',0)$
in the convex set 
$\cD_F$.
Therefore 
$(u,w(u))\in\cD_F^\circ$.
Equality~\eqref{Eq:F_derivatives}
follows by the chain rule. This proves the claim.

\noindent \textbf{Claim 3.} 
The following holds
for any strictly increasing sequence
$(u_n)_{n\in\bbN}$
with limit
$u_+$:
\begin{enumerate}[(a)]
\item \label{Claim3:a} if 
$u_+ = \sup \cJ = \sup \cI$, then
$$|w'(u_n)| \to \infty \qquad\text{as}\quad n \to \infty;$$ 
\item \label{Claim3:b} if 
$u_+ = \sup \cJ < \sup \cI$,
then
$$\|\nabla F(u_n,w_n)\| \to \infty \qquad\text{as}\quad n \to \infty.$$ 
\end{enumerate}

To prove the claim, assume that the conclusion of~\eqref{Claim3:a}
does not hold. Since the sequence 
$(w'(u_n))_{n\in\bbN}$
is non-decreasing by Lemma~\ref{Lem:existence},
there exists a finite positive number, denoted by
$w'(u_+)>0$,
such that
$\lim_{n\to\infty}w'(u_n)=w'(u_+)$.
Claim~2\eqref{Item:a},
applied to
$u=u_n$,
implies 
$(u_n,w_n)\in\cD_R^\circ$
for all 
$n\in\bbN$
and hence by~\eqref{eq:Important_Reduction}
$(u_+,w_+)$
is in the closure of 
$\cD_R$.
However
$(u_+,w_+)$
cannot be in the boundary of 
$\cD_R$
since
$R$
explodes at the boundary by assumption and it holds
$\lim_{n\to\infty}R(u_n,w_n)=0$
(recall that $R(u_n,w_n)=0$ for all 
$n\in\bbN$).
Therefore 
$(u_+,w_+)\in\cD_R^\circ$.
The derivatives
$\partial_1 R, \partial_2 R$
are hence continuous at 
$(u_+,w_+)$
and, in the limit
as 
$n\to\infty$,
formula~\eqref{Eq:R_derivatives}
and the fact
$w'(u_+)>0$
imply
\begin{equation*} 
\partial_2 R(u_+,w_+) = -\frac{\partial_1 R(u_+,w_+)}{w'(u_+)}.
\end{equation*}
Therefore either $0=\partial_1 R(u_+,w_+)=\partial_2 R(u_+,w_+)$ 
or both partial derivatives 
at
$(u_+,w_+)$
are non-zero. Suppose the former.
For an arbitrary 
$u \in (0,1)$,
the convexity of $R$ yields
\[- R(u,0) = R(u_+,w_+) - R(u,0) \le \nabla R(u_+,w_+) \cdot (u_+ - u, w_+)' = 0.\]
Since $R(u,0) < 0$ (see assumptions~A3~and~A4), this leads to a contradiction. 
Hence $\partial_1 R(u_+,w_+)$ and $\partial_2 R(u_+,w_+)$ are non-zero and 
related by the equality above.
By the implicit 
function theorem there exists 
an open interval
$N$ 
containing 
$u_+$ 
and a function 
$\wt{w}:N\to\bbR $,
such that 
$R(u,\wt{w}(u)) = 0$ 
for all 
$u\in N$. 
This contradicts the maximality 
of $\cI$ and 
proves Claim~3\eqref{Claim3:a}.
Note that under assumption of Claim~3\eqref{Claim3:b},
$(u_+,w_+)$ must be a boundary point of $\cD_F$. 
Since $F$ is steep this implies 
$\|\nabla F(u_n,w_n)\| \to \infty$
and the claim follows.

Theorem~\ref{thm:LCGFh} follows easily if
$\|\nabla F(u_n,w_n)\| \to \infty$ as
$n\to\infty$.
Indeed,
assumption~\eqref{eq:Important_Reduction}
and Lemma~\ref{Lem:existence}
imply that the sequence
$w'(u_n)>0$
is strictly increasing and positive
for all large 
$n$.
Since 
$F(0,0)=F(1,0)=0$,
Proposition~\ref{prop:Properties_F_and_R}~(B)
implies that 
$\partial_1 F(1,0)\geq0$.
The L\'evy-Khintchine representation of 
$F$
in~\eqref{Eq:FR_form}
implies
$\partial_2 F(u,w)\geq0$
for all
$(u,w)\in\cD_F$.
Since the gradient 
of the convex function 
$F$
is \emph{monotone}
on 
$\cD_F^\circ$
and 
$(u_n,w_n),(1,0)\in\cD_F^\circ$
for all 
$n$,
we find 
$$\partial_1 F(u_n,w_n) (u_n-1) + \partial_2 F(u_n,w_n) w_n \ge \partial_1 F(1,0)(u_n-1) + \partial_2
F(1,0) w_n\geq 0.$$
Therefore by~\eqref{Eq:F_derivatives}  we obtain
\begin{equation}\label{Eq:h_bound}
h'(u_n) 
\ge \partial_2 F(u_n,w_n)
\left(w'(u_n) - \frac{w_n}{u_n-1}\right).
\end{equation}
If
$|\partial_2 F(u_n,w_n)| \to \infty$ as
$n\to\infty$,
the steepness of 
$h$
at
$u_+$
follows from~\eqref{eq:Important_Reduction},~\eqref{Eq:h_bound}
and the fact that 
$w'(u_n)$
is strictly positive and increasing.
If
$|\partial_1 F(u_n,w_n)| \to \infty$ as
$n\to\infty$,
then,
since
$\partial_2 F\geq0$
on
$\cD_F^\circ$,
formula~\eqref{Eq:F_derivatives}
implies
Theorem~\ref{thm:LCGFh}.

If $\|\nabla F(u_n,w_n)\|$ does not tend to infinity as
$n\to\infty$,
the following facts hold:
$(u_+,w_+)\in\cD_F^\circ$
(since
$(u_+,w_+)$
is in the closure of 
$\cD_F$
and 
$F$
is steep), 
$|w'(u_n)| \to \infty$ as
$n \to \infty$
(by Claim~3) 
and
$\partial_2 F(u,w)\geq0$
for all
$(u,w)\in\cD_F^\circ$
(by L\'evy-Khintchine representation~\eqref{Eq:FR_form}
of 
$F$).
The next claim plays a key role in the proof 
of steepness of
$h$.\\ 
\noindent \textbf{Claim 4.} 
If $\|\nabla F(u_n,w_n)\|$ does not tend to infinity as
$n\to\infty$,
then
$\partial_2 F(u_+,w_+)>0$.
In particular there exists 
$\delta>0$
such that
$\partial_2 F(u_n,w_n)>\delta$
for all large
$n\in\bbN$.

Note first that
$\partial_2 F(u_+,w_+)$
is well-defined since
$(u_+,w_+)\in\cD_F^\circ$.
If
$\partial_2 F(u_+,w_+)=0$,
differentiation under the integral 
in~\eqref{Eq:FR_form}
implies that
$b_2=0$
and the support of 
$m$
is contained in the set
$\bbR\times\{0\}$.
This would imply that 
$F(0,w)=0$
for all 
$w\in\bbR$,
which contradicts the assumption in the theorem.
Hence the claim follows.

To conclude the proof of Step~(I), 
it remains to note that equality~\eqref{Eq:F_derivatives}
applied at
$u_n$
together with Claim~4
yield the steepness of 
$h$
in the case 
$\|\nabla F(u_n,w_n)\|$ does not tend to infinity.

\medskip

\noindent \textit{Step (II)}. We now prove that for any
$u\in \bbR\setminus \cJ$,
the limit in~\eqref{eq:DefOfh}
is equal to 
$+\infty$. This will conclude the proof of Theorem~\ref{thm:LCGFh}.

Let $t_n \downarrow 0$ and define 
$h_n(u) = \frac{1}{t_n}\log \EE{\E^{uX_{t_n}}}$ 
for all 
$u \in
\bbR$. 
We know that 
$\lim_{n\to\infty}h_n(u) =h(u)$ 
for all
$u\in\cJ$.
Moreover by Step~(I),
$h(u)$ is steep at the boundary of $\cJ$
and
$0\in\cJ^\circ$. 
Since $X_t$ is infinitely divisible for all $t \ge 0$
(see~\cite[Theorem 2.15]{Schachermayer}), 
there exist random variables 
$\wh{X}_n$ 
such that 
$h_n(u) = \log \EE{\E^{u\wh{X}_{n}}}$ 
(i.e. $h_n$ is the cumulant generating function of 
$\wh{X}_n$).  
Therefore there exists 
a random variable $X$ such that 
$\wh{X}_n \to X$ 
in distribution and,
if we define 
$H(u) = \log \EE{\E^{uX}}$
for all
$u \in \bbR$, 
the equality
$H(u) = h(u)$ holds on $\cJ$. 
Since 
$H$
is a cumulant generating function, 
it is lower semicontinuous and convex, and in particular 
continuously differentiable in the interior of 
its effective domain $\cD_H$. 
But $h$ is steep 
and hence non-differentiable at the boundary of 
$\cJ$. 
Therefore it follows that
$\cD_H = \cJ$ 
and 
$H(u) = \infty$ for all 
$u \in \bbR \setminus \cJ$. 
However for all 
$u \in \bbR \setminus \cJ$,
the Skorokhod representation theorem 
and Fatou's lemma imply 
\[\liminf_{n \to \infty} \EE{\E^{u \wh{X}_n}} \ge \EE{\E^{u X}} = \E^{H(u)} = \infty.\]
Hence the equality
$\lim_{n \to \infty} \frac{1}{t_n}\log \EE{\E^{uX_{t_n}}} = \infty$ 
holds for 
$u \in \bbR \setminus \cJ$. This concludes the proof of Theorem~\ref{thm:LCGFh}.~\exit

\subsection{Degenerate affine stochastic volatility models}
\label{subsec:Levy}

The remark following
Definition~\ref{Def:ASVM}
implies that in the case of a degenerate affine stochastic volatility
process
$(X,V)$,
the model 
$S=\E^X$
is an exponential L\'evy model
(note also that~A5 in 
this setting fails).
Therefore 
Definition~\ref{Def:ASVM}
and~\eqref{Eq:FR_form}
imply that the characteristic exponent 
$h(u):=F(u,0)$
of 
$X$
possesses a L\'evy-Khintchine 
characteristic triplet
$(\delta,\sigma^2,\nu)$,
where
$\delta,\sigma\in\bbR$
and
$\nu$
a L\'evy measure on
$\bbR\setminus\{0\}$,
and satisfies
\begin{eqnarray}
\label{eq:CharExpLevy}
h(u) & =  &
\log \EE\left[\exp\left(uX_1\right)|X_0=0\right]\\
& = & u\delta+\frac{1}{2}\sigma^2u^2+\int_{\bbR\setminus\{0\}}
\left(\E^{u\xi_1}-1-u\frac{\xi_1}{\xi_1^2+1}\right)\nu(\dd\xi_1)
\nonumber
\end{eqnarray}
for all
$u\in\bbC$
where the expectation 
exists.
The independence and stationarity of the increments of  
$X$
imply that 
$S$
is a martingale if and only if
$h(1)=0$,
which is, in terms of the characteristic triplet
$(\delta,\sigma^2,\nu)$,
equivalent to 
$ \int_{(1,\infty)}\E^{\xi_1}\nu(\dd\xi_1)<\infty$
and
\begin{eqnarray}
\label{eq:Exp_Levy_Drift}
 \delta=-\frac{1}{2}\sigma^2-\int_{\bbR\setminus\{0\}}
\left(\E^{\xi_1}-1-\frac{\xi_1}{\xi_1^2+1}\right)\nu(\dd\xi_1). 
\end{eqnarray}

The limiting cumulant generating function for the family
of random variables
$(X_t/t)_{t\geq1}$,
defined by the limit in~\eqref{eq:0_in_Domain_Assumption},
is in the case 
when 
$X$
is a L\'evy process
given trivially by 
$h$
in~\eqref{eq:CharExpLevy}, 
which therefore also coincides with
definition~\eqref{eq:DefOfh}.
The martingale condition for
$S=\E^X$
and the convexity of 
$h$
imply that 
$[0,1]$
is contained in the effective domain 
$\cD_h$.
In the case of affine stochastic volatility models
we had to establish Theorem~\ref{thm:LCGFh} to obtain 
sufficient condition for the set
$\{0,1\}$
to be contained in the interior 
$\cD_h^\circ$
of the effective domain
of $h$.
In the setting of L\'evy processes 
it is  well known 
(see e.g.~\cite[Theorem 25.17]{Sato})
that
$\{0,1\}\subset\cD_h^\circ$
if and only if 
\begin{eqnarray}
\label{eq:Eff_Dom_Levy}
\int_{(-\infty,-1)}\E^{u_0\xi_1}\nu\left(\dd \xi_1\right) +
\int_{(1,\infty)}\E^{u_1\xi_1}\nu\left(\dd \xi_1\right) <\infty\qquad
\text{for some}\quad
u_0<0\quad\text{and}\quad
u_1>1.
\end{eqnarray}
Condition~\eqref{eq:Eff_Dom_Levy}
implies that the interior of the effective domain
of 
$h$
is of the form
$\cD_h^\circ=(u_-,u_+)$
for some
$u_-\in[-\infty,0)$
and
$u_+\in(1,\infty]$.
It is therefore clear that 
$h$
is steep if and only if
\begin{eqnarray}
\label{eq:Lambda_Steep}
\int_{(-\infty,-1)}|\xi_1|\E^{\xi_1u_-}\nu\left(\dd \xi_1\right)=\infty &\text{and}& 
\int_{(1,\infty)}\xi_1\E^{\xi_1u_+}\nu\left(\dd \xi_1\right) =\infty,
\end{eqnarray}
where the integrals are taken to be infinite 
if the integrands take infinite value for some finite
$\xi_1$
(e.g. if 
$u_-=-\infty$
or
$u_+=\infty$).
Note also that
under assumption~\eqref{eq:Exp_Levy_Drift},
the L\'evy process 
$X$
is non-constant if and only if 
there is a Brownian component
(i.e. $\sigma^2>0$)
or its paths are discontinuous 
(i.e.
$\nu\neq0$).
Hence the equality
$$
h''(u) = \sigma^2 +
\int_{\bbR\setminus\{0\}}\xi_1^2\,\E^{u\xi_1}\,\nu(\dd\xi_1),
\qquad u\in\cD_h^\circ,
$$
implies
$h''(u)>0$
for all
$u\in\cD_h^\circ$. 
These arguments therefore imply 
Proposition~\ref{prop:OptionPricesLevy}, which
is the analogue of Theorem~\ref{thm:LCGFh}
for L\'evy processes. 

\begin{proposition}
\label{prop:OptionPricesLevy}
Let 
$X$
be a non-constant L\'evy process 
(i.e. 
the first component of a degenerate affine stochastic volatility process) 
with state-space 
$\bbR$,
characteristic triplet
$(\delta,\sigma^2,\nu)$
and the characteristic exponent 
$h$
given by~\eqref{eq:CharExpLevy}. 
Assume further that conditions~\eqref{eq:Exp_Levy_Drift},~\eqref{eq:Eff_Dom_Levy}
and~\eqref{eq:Lambda_Steep}
are satisfied.
Then the interior 
$\cD_h^\circ$
of the effective domain 
of
$h$
is an interval
$(u_-,u_+)$,
where
$u_-\leq u_+$
are extended real numbers, 
$h$
is a convex essentially smooth limiting cumulant generating function for the
family
$(X_t/t)_{t\geq1}$
and the set
$\{0,1\}$
is contained in the interior of 
$\cD_h$.
Furthermore,
$h$
is smooth on 
$\cD_h^\circ$
and
$h''(u)>0$
for all
$u\in\cD_h^\circ$.
\end{proposition}

\section{Rate functions and the option prices far from maturity}
\label{sec:OptionPrices}
In this section 
we describe the limiting behaviour of 
a family of European options under an affine stochastic 
volatility model
$S=\E^X$. 
These results will be used in Section~\ref{sec:App}
to prove the formulae for the limiting implied volatility 
smile.

In order to understand the limits of the vanilla option prices
far from maturity
in an affine stochastic volatility model
$S$,
we will need to apply the large deviation
principle for the family
$(X_t/t)_{t\geq1}$
under a risk-neutral measure 
$\PP$
\textit{and} under
the measure 
$\widetilde{\PP}$,
known as the \textit{share measure}.\footnote{The name
stems from the fact that under 
$\widetilde\PP$ the numeraire asset is the risky security
$S=\E^X$.}
Recall that 
for every
$t\geq0$
the measure 
$\widetilde{\PP}$
is equivalent to 
$\PP$
on the 
$\sigma$-field
$\cF_t$
and the Radon-Nikodym derivative is given by
$$
\frac{\dd \widetilde{\PP}}{\dd\PP}\Big\lvert_{\cF_t} \>= \> \E^{X_t}.
$$
The limiting cumulant generating function 
for
$(X_t/t)_{t\geq0}$
under 
$\widetilde\PP$
is defined by
\begin{eqnarray}
\label{eq:h_tilede_Def} \wt h(u) & :=  & \lim_{t\to\infty}\frac{1}{t}\log \wt\EE\left[\E^{uX_t}\right].
\end{eqnarray}
The function 
$\wt h$
and its effective domain 
$\cD_{\widetilde h}$
satisfy
\begin{eqnarray}
\label{eq:hTilda_Def}
\wt h(u) = h(u+1)\qquad\text{and}\qquad\cD_{\wt h}=\{u\in\bbR\>:\>1+u\in\cD_h\},  
\end{eqnarray} 
where 
$h$
is 
the limiting cumulant generating
function defined in~\eqref{eq:DefOfh}
and
$\cD_h$
is its effective domain.
Note that 
$0\in\cD_{\wt h}^\circ$
if and only if 
$1\in\cD_{h}^\circ$.
The identity in~\eqref{eq:hTilda_Def}
implies the following relationship between 
the Fenchel-Legendre transforms (see~\eqref{eq:DefFenchelLegendreTransf} 
for the definition) of 
$h$
and
$\wt h$:
\begin{equation}
\label{eq:LambaStarInPTilde}
\widetilde{h}^*(x)=h^*(x)-x\quad \text{for all}\quad x\in\bbR.
\end{equation}


Theorem~\ref{thm:OptionPrices} below describes
the limiting behaviour of certain European derivatives under 
an affine stochastic volatility process 
$(X,V)$.
Before we state it, we collect the following facts.

\begin{remarks}
\noindent (i) If
$(X,V)$
is a non-degenerate affine stochastic volatility process that satisfies assumptions
%
of Lemma~\ref{lem:h_not_zero},
then
the limiting cumulant generating functions
$h$
and
$\wt h$
(defined in~\eqref{eq:DefOfh}
and~\eqref{eq:h_tilede_Def}
respectively)
are strictly convex with strictly positive second derivatives
in the interior of their respective effective domains.
Remark~(B) after Definition~\ref{def:Essential_Smoothness}
implies that         
their convex duals
$h^*$
and
$\wt h^*$ 
are strictly convex and differentiable with respective unique global minima 
attained at
\begin{eqnarray}
\label{eq:GlobalMinima}
x^* = h'(0) &\text{and} & 
\wt x^* = \wt h'(0) = h'(1).
\end{eqnarray}
Lemma~\ref{lem:h_not_zero}
also implies
the following inequalities:
\begin{eqnarray}
\label{eq:GlobalMinima_Inequalities}
x^*& <\quad 0\quad < & \wt x^*. 
\end{eqnarray}
\noindent (ii) If
$(X,V)$
is a degenerate affine stochastic volatility process that satisfies assumptions
of Proposition~\ref{prop:OptionPricesLevy},
then
$h'$
is strictly increasing on 
$\cD_h^\circ$
and its image is equal to 
$\bbR$.
The unique
global minima 
of the Fenchel-Legendre transforms
$h^*$
and
$\wt h^*$
are (by Remark~(B) after Definition~\ref{def:Essential_Smoothness})
explicitly given by
\begin{subequations}
\label{eq:MinimumLevyFLTransform}
\begin{align}
x^* & =  h'(0) = -\frac{1}{2}\sigma^2-
\int_{\bbR\setminus\{0\}}\left(\E^{\xi_1}-1-\xi_1\right)\,\nu(\dd\xi_1),\\
\wt x^* & =  h'(1) = \frac{1}{2}\sigma^2 +
\int_{\bbR\setminus\{0\}}\left(\E^{\xi_1}(\xi_1-1)+1\right)\,\nu(\dd\xi_1).
\end{align}
\end{subequations}
Formulae~\eqref{eq:MinimumLevyFLTransform} 
show that the inequalities in~\eqref{eq:GlobalMinima_Inequalities}
hold also in the degenerate case.
\\
\noindent (iii) In the case of the Black-Scholes model 
(i.e. 
$\nu=0$),
the assumptions of Proposition~\ref{prop:OptionPricesLevy}
are satisfied. 
The effective domains of 
$h_{\mathrm{BS}}$
and
$\wt h_{\mathrm{BS}}$
are equal to 
$\bbR$ 
and the following 
formulae hold
\begin{eqnarray}
h_{\mathrm{BS}}(u)  =  \frac{1}{2}\sigma^2(u^2-u)& \text{and} &
\wt h_{\mathrm{BS}}(u)  =  \frac{1}{2}\sigma^2(u^2+u) 
\quad\text{for}\quad u\in\bbR,\\
h^*_{\mathrm{BS}}\left(x;\sigma^2\right)  = \frac{1}{2\sigma^2} \left(x+\frac{\sigma^2}{2}\right)^2
& \text{and} &
\wt h^*_{\mathrm{BS}}\left(x;\sigma^2\right)  = \frac{1}{2\sigma^2}
\left(x-\frac{\sigma^2}{2}\right)^2
\quad\text{for}\quad x\in\bbR.
\label{eq:BS_Polynomial}
\end{eqnarray}
Therefore we have
$x^*=-\sigma^2/2$
and
$\wt x^*=\sigma^2/2$.
\end{remarks}

\begin{theorem}
\label{thm:OptionPrices} 
Let 
$(X,V)$
be a non-degenerate (resp. degenerate) 
affine stochastic volatility process that
satisfies the assumptions of Theorem~\ref{thm:LCGFh}
or Corollaries~\ref{cor:LCGFh}~(i),~\ref{cor:LCGFh}~(ii)
(resp. Proposition~\ref{prop:OptionPricesLevy}).
Then the family of random variables 
$(X_t/t)_{t\geq1}$
satisfies the LDPs 
under the measures
$\PP$
and
$\widetilde{\PP}$
with the respective good rate functions
$h^*$
and
$\widetilde{h}^*$,
where
$h$
is given in~\eqref{eq:DefOfh}
(resp.~\eqref{eq:CharExpLevy})
and 
$\widetilde{h}$
in~\eqref{eq:h_tilede_Def}.
Fix
$x\in\bbR$,
let
$x^*,\wt x^*$
be as in~\eqref{eq:GlobalMinima}
(resp.~\eqref{eq:MinimumLevyFLTransform}) 
and denote
$S=\E^{X}$
and
$y^+:=\max\{0,y\}$
for 
$y\in\bbR$.
\begin{enumerate}
\item[(i)] 
The asymptotic behaviour of a put option
with strike
$S_0\E^{xt}$
is given by the following formula
\begin{eqnarray*}
\label{eq:call}
\lim_{t\to\infty}t^{-1}\log\EE\left[\left(S_0\E^{xt}-S_t\right)^+\right]  & =   &
\left\{
\begin{array}{ll}
x-h^*\left(x\right)\quad & \text{if }x\leq x^*,\\
x\quad & \text{if }x>x^*.
\end{array}
\right.
\end{eqnarray*}
\item[(ii)] 
The asymptotic behaviour of a call option,
struck at 
$S_0\E^{tx}$,
is given by the formula
\begin{eqnarray*}
\lim_{t\to\infty}t^{-1}\log\EE\left[\left(S_t-S_0\E^{xt}\right)^+\right]  & =  & 
\left\{
\begin{array}{ll}
-\wt h^*\left(x\right)\quad & \text{if }x\geq \wt x^*,\\
0\quad & \text{if }x<\wt x^*. 
\end{array}
\right.
\end{eqnarray*}
\item[(iii)] 
The asymptotic behaviour of a covered call option
with payoff
$S_t-(S_t-S_0\E^{tx})^+$
is given by 
\begin{eqnarray*}
\label{eq:CoveCall}
\lim_{t\to\infty}t^{-1}\log\left(S_0-\EE\left[\left(S_t-S_0\E^{xt}\right)^+\right]\right) & =  & 
\left\{
\begin{array}{ll}
0\quad & \text{if }x> \wt x^*,\\
x-h^*\left(x\right)\quad  &\text{if }x\in\left[x^*,\wt x^*\right],\\
x\quad & \text{if }x< x^*. 
\end{array}
\right.
\end{eqnarray*}
\end{enumerate}
Furthermore the convergence in
(i)-(iii)
is uniform in 
$x$
on compact subsets of 
$\bbR$.
\end{theorem}

\begin{remarks} (I) The formulae in (i), (ii) and (iii) of 
Theorem~\ref{thm:OptionPrices}
are continuous in 
$x$
since
the value of 
the Fenchel-Legendre transforms 
$h^*$
(resp. 
$\wt h^*$)
at
$x^*$
(resp. 
$\wt x^*$)
is equal to zero.
Note further that the formulae in 
Theorem~\ref{thm:OptionPrices}
are independent of the starting value
$(X_0,V_0)$
of the 
model.  \\
\noindent (II) The reason for studying the limiting behaviour of the put, call and 
covered call in Theorem~\ref{thm:OptionPrices}
lies in the fact that these payoffs yield non-trivial limits
on complementary subintervals of
$\bbR$,
thus 
obtaining a non-trivial limit for every
$x\in\bbR$.
This limit will be
compared 
in Section~\ref{sec:App}
with the corresponding limit 
in the Black-Scholes
model, which will yield
the formula for the limiting implied volatility smile
under affine stochastic volatility models.
\end{remarks}

\begin{proof}
Assume first that 
$(X,V)$
is a non-degenerate affine stochastic volatility process.
The limiting cumulant generating function 
$h$
satisfies~\eqref{eq:0_in_Domain_Assumption}
and is essentially smooth either by Theorem~\ref{thm:LCGFh}
or by Corollaries~\ref{cor:LCGFh}~(i),~\ref{cor:LCGFh}~(ii).
Hence its convex dual 
$h^*$
is non-negative
(by~\eqref{eq:DefFenchelLegendreTransf} and the fact
$h(0)=0$),
has
compact level sets (by~\eqref{eq:Compac_LEvel_Sets}
and $0\in\cD_h^\circ$)
and is differentiable on 
$\cD_{h^*}=\bbR$
with strictly increasing first derivative
(by Lemma~\ref{lem:h_not_zero} and Remark~(B) following Definition~\ref{def:Essential_Smoothness}).
Therefore by Theorem~\ref{thm:GartnerEllis}
the family
$(X_t/t)_{t\geq1}$
satisfies the LDP under
$\PP$
with the good rate function
$h^*$.
Since
$1\in\cD_h^\circ$,
by~\eqref{eq:hTilda_Def}
the function
$\wt h$
satisfies 
the condition in~\eqref{eq:0_in_Domain_Assumption}.
Therefore 
all of the assumptions 
of Theorem~\ref{thm:GartnerEllis}
hold
under
$\wt \PP$
and hence
$(X_t/t)_{t\geq1}$
satisfies the LDP 
with the good rate function
$\wt h^*$.
Furthermore
$\wt h^*$
enjoys the same regularity
on 
$\cD_{\wt h^*}=\bbR$
as the rate function
$h^*$.
The LDPs in the degenerate case follow 
from the same argument with 
Theorem~\ref{thm:LCGFh}, Corollaries~\ref{cor:LCGFh}~(i),~\ref{cor:LCGFh}~(ii)
and Lemma~\ref{lem:h_not_zero}
replaced by Proposition~\ref{prop:OptionPricesLevy}.

We now prove the formulae in Theorem~\ref{thm:OptionPrices}.
Without loss of generality we may assume that 
$S_0=1$, i.e.
$X_0=0$.
The following inequality holds for all 
$t\geq1$
and
$\varepsilon>0$:
$$
\E^{tx}\left(1-\E^{-\varepsilon}\right)
I_{\{X_t/t< x-\varepsilon\}}
\leq 
\left(\E^{xt}-\E^{X_t}\right)^+
\leq 
\E^{tx}
I_{\{X_t/t< x\}}.
$$
Hence by taking expectations, logarithms, multiplying by
$1/t$
and applying the LDP for 
$(X_t/t)_{t\geq1}$
under
$\PP$
we obtain the following inequalities
$$
x-\inf_{y<x-\varepsilon}h^*(y)\leq 
\liminf_{t\to\infty}
\frac{1}{t}\log\EE\left[\left(\E^{xt}-\E^{X_t}\right)^+\right]\leq
\limsup_{t\to\infty}
\frac{1}{t}\log\EE\left[\left(\E^{xt}-\E^{X_t}\right)^+\right]\leq
x-\inf_{y\leq x}h^*(y).
$$
Since 
$h^*$
is continuous 
on 
$\bbR$,
strictly decreasing
for 
$x\leq x^*$
and takes value 
$0$
at
$x^*$,
the formula in Theorem~\ref{thm:OptionPrices}~(i)
holds.

We now consider the call option case.
The following inequality holds for all 
$t\geq1$
and
$\varepsilon>0$:
$$
\E^{X_t}\left(1-\E^{-\varepsilon}\right)
I_{\{X_t/t> x+\varepsilon\}}
\leq 
\left(\E^{X_t}-\E^{xt}\right)^+
\leq 
\E^{X_t}
I_{\{X_t/t> x\}}.
$$
Again by taking expectations, changing measure to 
$\widetilde{\PP}$,
applying logarithms, multiplying by
$1/t$
and applying the LDP for 
$(X_t/t)_{t\geq1}$
under
$\widetilde{\PP}$
we obtain the following inequalities
$$
-\inf_{y>x+\varepsilon}\widetilde{h}^*(y)\leq 
\liminf_{t\to\infty}
\frac{1}{t}\log\EE\left[\left(\E^{X_t}-\E^{xt}\right)^+\right]\leq
\limsup_{t\to\infty}
\frac{1}{t}\log\EE\left[\left(\E^{X_t}-\E^{xt}\right)^+\right]\leq
-\inf_{y\geq x}\widetilde{h}^*(y).
$$
Note that 
$\wt x^*$
is a global minimum 
for 
$\wt h^*$
at which value $0$
is attained.
The continuity of 
$\widetilde{h}^*$
implies 
the formula in Theorem~\ref{thm:OptionPrices}~(ii).

In the case of the covered call, 
the following simple inequalities hold for all
$x\in\bbR$:
\begin{eqnarray}
\label{eq:between}
\E^{xt}I_{\{X_t/t\geq x\}}
& \leq\quad
\E^{X_t}-\left(\E^{X_t}-\E^{tx}\right)^+\quad= &
\E^{X_t}I_{\{X_t/t< x\}}+
\E^{xt}I_{\{X_t/t\geq x\}},\\
\label{eq:Above}
\E^{X_t}I_{\{X_t/t\leq x\}}
&
\leq\quad
\E^{X_t}-\left(\E^{X_t}-\E^{tx}\right)^+
\quad \leq&
\E^{X_t},\\
\label{eq:Below}
\E^{xt}I_{\{X_t/t\geq x\}}
& \leq \quad
\E^{X_t}-\left(\E^{X_t}-\E^{tx}\right)^+
\quad \leq &
\E^{xt}.
\end{eqnarray}
Inequality~\eqref{eq:between}
and the LDP under measures
$\PP$
and
$\wt \PP$ 
imply
the inequalities
\begin{eqnarray*}
\E^{xt}\PP\left[X_t/t\geq x\right]& \leq &
1-\EE\left[\left(\E^{X_t}-\E^{xt}\right)^+\right]\quad
=\quad
\widetilde{\PP}\left[X_t/t< x\right]+
\E^{xt}\PP\left[X_t/t\geq x\right]\\
&\leq& 
\exp\left(-t\inf_{y\leq x}\widetilde{h}^*(y)+\varepsilon t\right)+
\E^{xt}\exp\left(-t\inf_{y\geq x}h^*(y)+\varepsilon t\right)
\end{eqnarray*}
for any 
$x\in\bbR$,
$\varepsilon>0$
and 
$t$
large enough.
Assume now 
$x\in\left[x^*,\wt x^*\right]$
and note  that in this case we have
$\inf_{y\geq x}h^*(y)=h^*(x)$
and
$\inf_{y\leq x}\wt{h}^*(y)=\wt h^*(x)$.
By~\eqref{eq:LambaStarInPTilde}
we obtain
$$
x+\frac{1}{t}\log\PP\left[X_t/t\geq x\right]\leq
\frac{1}{t}\log\left(1-\EE\left[\left(\E^{X_t}-\E^{xt}\right)^+\right]\right)
\leq x-h^*(x)+\varepsilon+\frac{1}{t}\log2
$$
for any 
$\varepsilon$
and all large
$t$.
Therefore we find the inequalities 
\begin{eqnarray*}
x-h^*(x)
&\leq&
\liminf_{t\to\infty}\frac{1}{t}\log\left(1-\EE\left[\left(\E^{X_t}-\E^{xt}\right)^+\right]\right)\\
&\leq&
\limsup_{t\to\infty}\frac{1}{t}\log\left(1-\EE\left[\left(\E^{X_t}-\E^{xt}\right)^+\right]\right)
\leq x-h^*(x)+\varepsilon
\end{eqnarray*}
for all
$\varepsilon>0$.
This proves the formula in Theorem~\ref{thm:OptionPrices}~(iii)
for 
$x\in[x^*,\wt x^*]$.

Assume that
$x>\wt x^*$
and 
take expectations, change measure to 
$\widetilde{\PP}$,
apply logarithms and multiply by
$1/t$
the inequalities in~\eqref{eq:Above}
to obtain the following: 
$$
\frac{1}{t}\log\wt\PP\left[X_t/t\leq x\right]\leq
\frac{1}{t}\log\left(1-\EE\left[\left(\E^{X_t}-\E^{xt}\right)^+\right]\right)
\leq 0. 
$$
Since 
$\inf_{y<x} \wt h^*(y) = 0$
the LDP for 
$(X_t/t)_{t\geq1}$
under
$\widetilde{\PP}$
implies the formula in Theorem~\ref{thm:OptionPrices}~(iii)
that corresponds to 
$x>\wt x^*$.

Finally let
$x<x^*$.
Inequalities in~\eqref{eq:Below}
imply the following
$$
x+\frac{1}{t}\log\PP\left[X_t/t\geq x\right]\leq
\frac{1}{t}\log\left(1-\EE\left[\left(\E^{X_t}-\E^{xt}\right)^+\right]\right)
\leq x. 
$$
An application of the LDP for 
$(X_t/t)_{t\geq1}$
under
$\PP$
completes the proof of part~(iii).

We now show that the limits in the theorem are uniform
in
$x$
on compact sets in 
$\bbR$.
Since the argument is similar in all the cases, we concentrate on
Theorem~\ref{thm:OptionPrices}~(i).
Let
$(x_0,y_0)$
be a finite interval in
$\bbR$
and define for any
$x\in\bbR$
and
$t\geq1$
$$V(t,x)=
t^{-1}\log\EE\left[\left(\E^{xt}-\E^{X_t}\right)^+\right] - v(x),
$$
where 
$v(x)$
denotes the continuous limit that appears 
in Theorem~\ref{thm:OptionPrices}~(i).
It follows that 
$$
V(t,x_0)+v(x_0)-v(x)
\leq V(t,x)
\leq
V(t,y_0)+v(y_0)-v(x) \qquad\text{for any}\quad x\in(x_0,y_0).
$$
We therefore find
\begin{eqnarray}
\label{eq:Ineq_Uniform_conv}
\lvert V(t,x) \rvert
&\leq& \max\left\{ \lvert V(t,y_0)\rvert,
\lvert V(t,x_0) \rvert
\right\} 
+
\max\left\{
\lvert v(x)-v(x_0)\rvert,
\lvert v(x)-v(y_0)\rvert
\right\}.
\end{eqnarray}
Since we have already proved that 
$\lim_{t\to\infty}|V(t,x)|=0$
for every $x$
and the limiting function
$v(x)$
is continuous, and hence uniformly continuous on every compact
set, the inequality in~\eqref{eq:Ineq_Uniform_conv}
implies that the logarithms of the put option prices converge to 
$v(x)$
uniformly in 
$x$
on compact sets in 
$\bbR$.

\end{proof}

\section{Asymptotic behaviour of the implied volatility} 
\label{sec:App}
The value 
$C(S_0,K,t,\sigma^2)$
of the European call option with strike
$K$
and expiry 
$t$
in 
a Black-Scholes 
model
(i.e. degenerate affine stochastic volatility model
without jumps, see Section~\ref{subsec:Levy})
is given by 
the Black-Scholes formula
\begin{eqnarray}
\label{eq:BS_Formula}
C(S_0,K,t,\sigma^2) & = & S_0\,N(d_+)-K\,N(d_-),\qquad
\text{where}
\quad
d_{\pm}=\frac{\log(S_0/K)\pm\sigma^2t/2}{\sigma\sqrt{t}}
\end{eqnarray}
and
$N(\cdot)$
is the standard normal cumulative distribution function.
Let 
$S=\E^{X}$
be an affine stochastic volatility model
from Definition~\ref{Def:ASVM}
with the starting point 
$S_0=\E^{X_0}$.
The \textit{implied volatility}
in the model 
$S=\E^X$
for the strike 
$K>0$
and
maturity 
$t>0$
is the unique positive number
$\wh \sigma(K,t)$
that satisfies the following equation in the variable
$\sigma$:
\begin{eqnarray}
\label{eq:DefOfImpiedVol}
C\left(S_0,K,t,\sigma^2\right)=\EE\left[\left(S_t-K\right)^+\right].
\end{eqnarray}
Implied volatility is well-defined since 
the function
$\sigma\mapsto C\left(S_0,K,t,\sigma^2\right)$
is strictly increasing for 
positive
$\sigma$
(i.e. \textit{vega} of a call option
$\frac{\partial C}{\partial \sigma}(S_0,K,t,\sigma^2)
=S_0 N'(d_+)\sqrt{t}$
is strictly positive) and
the right-hand side of~\eqref{eq:DefOfImpiedVol} 
lies in the image of the Black-Scholes formula
by a no-arbitrage argument.
Put-call parity, which holds since 
$S=\E^X$
is a true martingale, implies the identity
$P\left(S_0,K,t,\wh \sigma(K,t)^2\right)=\EE\left[\left(K-S_t\right)^+\right]$,
where
$P\left(S_0,K,t,\sigma^2\right)$
denotes the price of the put option in the Black-Scholes model
with volatility 
$\sigma$.

If the affine stochastic volatility process
$(X,V)$
satisfies the assumptions of 
Theorem~\ref{thm:OptionPrices},
then the implied volatility has the following limit
\begin{eqnarray}
\label{eq:Fixed_Strike_Limit}
\lim_{t\to\infty}\wh \sigma(K,t)=2\sqrt{2 h^*(0)}=
2\sqrt{-2 h\left((h')^{-1}(0)\right)} 
\end{eqnarray}
for any fixed strike
$K>0$,
where 
$h^*$
is the rate function of the model
(the second equality in~\eqref{eq:Fixed_Strike_Limit} 
follows from~\eqref{eq:FromulaForTrnsform}
and~\eqref{eq:DerivativeOfRateFunction}).
Tehranchi~\cite{Tehr09} 
proved that the first equality in~\eqref{eq:Fixed_Strike_Limit}
holds uniformly 
in
$K$
on compact sets 
in
$\Rplus$
for non-negative local martingales 
with cumulant generating functions 
that satisfy certain additional conditions.
Note that the limit in~\eqref{eq:Fixed_Strike_Limit}
is independent of 
$K$,
which corresponds to the well-known 
flattening of the implied volatility smile at large maturities.
The uniform limit (in $K$)
on compact subsets 
of
$\Rplus$,
given in~\eqref{eq:Fixed_Strike_Limit},
also follows from Theorem~\ref{thm:AsymVol}
for affine stochastic volatility processes
(both in non-degenerate and degenerate, i.e. L\'evy, cases).

In order to obtain a non-trivial limit at infinity we define 
the implied volatility  
$\sigma_t(x)$
for the strike
$K=S_0\exp(tx)$,
where
$x\in\bbR$,
by
\begin{eqnarray}
\label{eq:Simga_t_Def}
\sigma_t(x) & =  & \wh \sigma\left(S_0\exp(tx),t\right).
\end{eqnarray}
We will show that if 
$(X,V)$
satisfies the assumptions of 
Theorem~\ref{thm:OptionPrices},
then the limiting implied volatility takes the form
\begin{eqnarray}
\label{eq:Def_Sigma_Infty}
 \sigma_\infty(x) & = & \sqrt{2}\left[\sgn(\wt x^*-x)\sqrt{\wt h^*(x)}+\sgn(x-x^*) \sqrt{h^*(x)}\right]
\qquad\text{for}\quad x\in\bbR,
\end{eqnarray}
where
$h^*$
and
$\wt h^*$
are the Fenchel-Legendre transforms (see~\eqref{eq:DefFenchelLegendreTransf} for definition)
of the limiting cumulant generating functions 
$h$
and
$\wt h$
of 
$X$
under
$\PP$
and
$\wt \PP$
respectively
and 
$x^*=h'(0)$,
$\wt x^*=\wt h'(0)=h'(1)$.
The function 
$\sgn(x)$
is  by definition equal to 
$1$
if 
$x\geq0$
and
$-1$
otherwise.

\begin{remarks} (i) 
Under the assumptions of 
Theorem~\ref{thm:OptionPrices},
the points 
$x^*$,
$\wt x^*$
are the locations of the unique global minima of the good rate functions
$h^*$
and 
$\wt h^*$
respectively and
by~\eqref{eq:GlobalMinima_Inequalities}
satisfy 
$x^*<0<\wt x^*$.
Note that 
$\wt h^*(x)\leq h^*(x)$
for 
$x\geq0$
and 
$\wt h^*(x)\geq h^*(x)$
for 
$x\leq0$
and hence 
the following strict inequality 
$ \sigma_\infty(x)>0$
holds
for all 
$x\in\bbR$.\\
\noindent (ii) The function $\sigma_\infty:\bbR\to(0,\infty)$
in~\eqref{eq:Def_Sigma_Infty}
is chosen so that it satisfies the following identities:
\begin{eqnarray}
\label{eq:ImportantIdentity}
h^*_{\mathrm{BS}}\left(x;\sigma_\infty(x)^2\right) = h^*(x)
\qquad\text{and}\qquad
\wt h^*_{\mathrm{BS}}\left(x;\sigma_\infty(x)^2\right) = \wt h^*(x),
\qquad x\in\bbR,
\end{eqnarray}
where the polynomials 
$h^*_{\mathrm{BS}}\left(x;\sigma^2\right)$
and
$\wt h^*_{\mathrm{BS}}\left(x;\sigma^2\right)$
are given in~\eqref{eq:BS_Polynomial}.  
Quantities of importance in the proof of 
Theorem~\ref{thm:AsymVol}
will be the following partial derivatives
\begin{eqnarray}
\label{eq:PartialDeriv_Lambda_BS}
\frac{\partial h^*_{\mathrm{BS}}}{\partial \sigma^2}\left(x;\sigma^2\right)  =
\frac{\partial\wt h^*_{\mathrm{BS}}}{\partial \sigma^2}\left(x;\sigma^2\right)  
= \frac{1}{8\sigma^4} \left(\sigma^2+2x\right)\left(\sigma^2-2x\right).
\end{eqnarray}
\noindent (iii) In formula~\eqref{eq:Def_Sigma_Infty},
the function 
$\sigma_\infty(x)$
is given as a linear combination of 
$\sqrt{h^*(x)}$
and
$\sqrt{\wt h^*(x)}$.
The coefficients in this linear combination are not 
uniquely determined by identities~\eqref{eq:ImportantIdentity}
(there are four possibilities).
However definition~\eqref{eq:Def_Sigma_Infty}
is the only choice that 
implies the following important properties
\begin{eqnarray}
\label{eq:Large}
 \sigma_\infty(x)^2 & < &
2|x|,\qquad \text{for}\quad x\in\bbR\setminus[x^*,\wt x^*]\quad\text{and}\quad
 \sigma_\infty(x^*)^2=2x^*,\quad
 \sigma_\infty(\wt x^*)^2=2\wt x^*,\\
\label{eq:Noraml_Size}
 \sigma_\infty(x)^2 & > & 2|x|,
\qquad 
\text{for}\quad
x\in(x^*,\wt x^*),
\end{eqnarray}
which will be crucial in the proof of Theorem~\ref{thm:AsymVol}.
Note that~\eqref{eq:Large} and~\eqref{eq:Noraml_Size}
trivially hold in the Black-Scholes model:
$x^*=-\sigma^2/2$,
$\wt x^*=\sigma^2/2$
and the limiting smile is constant
and equal to
$\sigma$.
The inequality in~\eqref{eq:Large}
on the interval 
$(-\infty,x^*)$
follows from the identity
$ \sigma_\infty(x)^2 + 2x = 4\left[h^*(x)-\sqrt{h^*(x)^2-xh^*(x)} \right]$
and the fact that
for 
$x<x^*$
we have
$h^*(x)>0$. 
Likewise the identity
$ \sigma_\infty(x)^2 + 2x = 4\left[h^*(x)+\sqrt{h^*(x)^2-xh^*(x)} \right]$
for 
$x\in(x^*,0)$
yields half of~\eqref{eq:Noraml_Size}.
The other half of~\eqref{eq:Large} and~\eqref{eq:Noraml_Size}
follow from analogous identities
involving
$\wt h^*$.\\
\noindent (iv) The arbitrary choice for
$\sgn(0)=1$
in~\eqref{eq:Def_Sigma_Infty}
is of no consequence since
$h^*(x^*)=\wt h^*(\wt x^*)=0$.
\end{remarks}

\begin{theorem}
\label{thm:AsymVol}
Let 
$(X,V)$
be 
an affine stochastic volatility
process that satisfies the assumptions of 
Theorem~\ref{thm:OptionPrices}.
Then we have
$$
\lim_{t\nearrow \infty}
\sigma_t(x) =  \sigma_\infty(x)
\quad\text{for any}\quad 
x\in\bbR,
$$
where
$ \sigma_t(x)$,
defined in~\eqref{eq:Simga_t_Def},
is the implied volatility in the model
$S=\E^X$
for the strike
$K=S_0\E^{xt}$,
and 
$\sigma_\infty(x)$
is given in~\eqref{eq:Def_Sigma_Infty}.
Furthermore for any
compact subset
$C$
in 
$\bbR\setminus\{x^*,\wt x^*\}$,
where 
$x^*,\wt x^*$
are defined in~\eqref{eq:GlobalMinima},
we have
$$
\sup_{x\in C}\>\Big\lvert \sigma_t(x) -  \sigma_\infty(x)
\Big\rvert
\longrightarrow 0
\quad\text{as}\quad t\nearrow \infty.
$$
\end{theorem}

\begin{remark}
Theorem~\ref{thm:AsymVol}
implies formula~\eqref{eq:Fixed_Strike_Limit}
obtained in~\cite{Tehr09}:
define
$x=\log(K/S_0)/t$
and apply the uniform convergence of 
$ \sigma_t(x)$
on a compact neighbourhood of zero.
\end{remark}

\begin{proof}
We can assume without loss of generality that
$S_0=1$.
Assume further that
$x_0>\wt x^*$
and pick 
$\varepsilon>0$.
The goal is to find
$\delta>0$
such that the following inequality 
holds for all large 
$t$:
\begin{eqnarray}
\label{eq:EpsilonInequality}
\lvert\sigma_t(x)^2-\sigma_\infty(x)^2 \rvert< \varepsilon,
\quad\text{where}\quad
x\in(x_0-\delta,x_0+\delta).
\end{eqnarray}
Inequality~\eqref{eq:Large}
implies that there exists 
$\varepsilon'\in(0,\varepsilon)$
such that 
$ \sigma_\infty(x_0)^2+\varepsilon' <2x_0$
and
$0< \sigma_\infty(x_0)^2-\varepsilon'$.
By~\eqref{eq:PartialDeriv_Lambda_BS}
we conclude that 
$\sigma^2\mapsto \wt h^*_{\mathrm{BS}}\left(x;\sigma^2\right)$
is strictly decreasing on the interval
$(0,2x)$
and hence obtain
the following inequalities:
$$
\wt h^*_{\mathrm{BS}}\left(x_0; \sigma_\infty(x_0)^2-\varepsilon'\right)\> >\>
\wt h^*_{\mathrm{BS}}\left(x_0; \sigma_\infty(x_0)^2\right)\> >\>
\wt h^*_{\mathrm{BS}}\left(x_0; \sigma_\infty(x_0)^2+\varepsilon'\right).
$$
Since all the functions are continuous and identitiy~\eqref{eq:ImportantIdentity}
holds, there exists a 
$\delta>0$
such that 
$x_0-\delta>\wt x^*$
and the strict inequalities hold
\begin{eqnarray}
\label{eq:Estimate_For_Lambda}
-\wt h^*_{\mathrm{BS}}\left(x; \sigma_\infty(x)^2-\varepsilon'\right)\> <\>
-\wt h^*(x) \><\>
-\wt h^*_{\mathrm{BS}}\left(x; \sigma_\infty(x)^2+\varepsilon'\right)
\end{eqnarray}
for all
$x\in (x_0-\delta,x_0+\delta)$.
Theorem~\ref{thm:OptionPrices}
implies that the call option converges uniformly 
on the interval
$(x_0-\delta,x_0+\delta)$
to
$-\wt
h^*(x)=\lim_{t\to\infty}t^{-1} \log\EE\left[\left(\E^{X_t}-\E^{xt}\right)^+\right]$.
In particular in the Black-Scholes model we get
$-\wt h^*_{\mathrm{BS}}\left(x; \sigma_\infty(x)^2\pm\varepsilon'\right)=
\lim_{t\to\infty}t^{-1} \log C(1,\E^{xt},t, \sigma_\infty(x)^2\pm\varepsilon')$
and the convergence is uniform
in
$x$
on 
$(x_0-\delta,x_0+\delta)$.
Since
$ \sigma_t(x)$
satisfies 
$\EE\left[\left(\E^{X_t}-\E^{xt}\right)^+\right]=C(1,\E^{xt},t, \sigma_t(x)^2)$
by definition, the inequalities 
in~\eqref{eq:Estimate_For_Lambda}
imply that 
$$
C(1,\E^{xt},t, \sigma_\infty(x)^2-\varepsilon')\> < \>
C(1,\E^{xt},t, \sigma_t(x)^2)\> < \>
C(1,\E^{xt},t, \sigma_\infty(x)^2+\varepsilon')
$$
for all 
$x\in(x_0-\delta,x_0+\delta)$
and all large
$t$.
Since the Black-Scholes formula is strictly increasing in 
$\sigma^2$
(i.e. vega is strictly positive), 
these inequalities imply~\eqref{eq:EpsilonInequality}.
This proves uniform convergence on any compact subset
$C$
of
$(\wt x^*,\infty)$.
The proof for a compact set
$C\subset(-\infty,\wt x^*)\setminus\{x^*\}$
is analogous.

We now consider convergence at the point
$\wt x^*$.
Pick any
$\varepsilon>0$
such that
$ \sigma_\infty(\wt x^*)^2=2\wt x^*>\varepsilon$.
Identity~\eqref{eq:BS_Polynomial}
implies that 
$$\wt h^*_{\mathrm{BS}}\left(\wt x^*; \sigma_\infty(\wt x^*)^2-\varepsilon\right)>
\wt h^*_{\mathrm{BS}}\left(\wt x^*; \sigma_\infty(\wt
x^*)^2\right)=\wt h^*(\wt x^*)=0 <
\wt h^*_{\mathrm{BS}}\left(\wt x^*; \sigma_\infty(\wt x^*)^2+\varepsilon\right).$$
The first inequality and the argument above imply that
$\sigma_\infty(\wt x^*)^2 - \varepsilon< \sigma_t(\wt x^*)^2$
for all large
$t$.
Since
$ \sigma_\infty(\wt x^*)^2+\varepsilon>2\wt x^*$,
the second inequality, 
Theorem~\ref{thm:OptionPrices}
yields
$$
1-C(1,\E^{\wt x^*t},t, \sigma_t(\wt x^*)^2)\> > \>
1-C(1,\E^{\wt x^*t},t, \sigma_\infty(\wt x^*)^2+\varepsilon)
$$
for all large
$t$.
This implies
$\sigma_t(\wt x^*)^2 <\sigma_\infty(\wt x^*)^2 + \varepsilon$
and hence proves the theorem for
$\wt x^*$.
The point
$x^*$
can be dealt with analogously. 
\end{proof}

The following corollary is a simple consequence of our results.

\begin{corollary}
\label{cor:LimitSmila}
Let 
$(X,V)$ 
be a non-degenerate affine stochastic volatility process that satisfies the assumptions
of Theorem~\ref{thm:OptionPrices}.
Then there exists a L\'evy process
$Y$,
which satisfies assumptions of 
Theorem~\ref{thm:AsymVol}
as a degenerate affine stochastic volatility process, 
such that the limiting smiles of the models
$\E^X$
and
$\E^Y$
are identical.
\end{corollary}

\begin{proof}
Let 
$h$
be the limiting cumulant generating function for
$(X,V)$.
Theorem~\ref{Thm:wm_convergence}
implies that 
$h$
is a cumulant generating function of an infinitely 
divisible random variable.
By Theorem~\ref{thm:LCGFh}, 
the characteristic triplet of 
$h$
satisfies conditions~\eqref{eq:Eff_Dom_Levy} and~\eqref{eq:Lambda_Steep}.
Therefore, if we define a L\'evy process 
$Y$
with this characteristic triplet, 
Theorem~\ref{thm:AsymVol}
and formula~\eqref{eq:Def_Sigma_Infty}
imply that models 
$X$
and
$Y$
have identical limiting volatility smiles.
\end{proof}

\begin{remarks} (i) In other words Corollary~\ref{cor:LimitSmila}
states that in the limit, non-degenerate affine stochastic volatility models  
cannot generate the behaviour of implied volatility, which is different from 
that generated by the processes with constant volatility and stationary, 
independent increments. \\ 
\noindent (ii) Corollary~\ref{cor:LimitSmila} suggests the following natural 
open question: can any limiting smile of an exponential L\'evy model
be obtained as a limit of implied volatility smiles of a non-degenerate 
affine stochastic volatility process? It is not immediately clear how to approach this
problem because the characterisation of the limiting cumulant generating function 
$h$
of a model
$(X,V)$,
given in Theorems~\ref{Thm:wm_convergence}
and~\ref{thm:LCGFh},
does not give an explicit form 
of L\'evy-Khintchine
triplet of
$h$.
\end{remarks}

\subsection{Examples of limiting smiles}
\label{subsec:PlotsOfSmiles}
We now apply the analysis to the examples of affine stochastic volatility models
described in Section~\ref{subsec:Examples}.
In each of the cases 
the limiting cumulant generating function
$h$
is available in
closed form.
If the assumptions
of 
Theorem~\ref{thm:LCGFh}
or
Corollaries~\ref{cor:LCGFh}~(i),~\ref{cor:LCGFh}~(ii)
are satisfied,
then
the convex dual
$h^*$
is a good rate function and hence the formula in~\eqref{eq:Def_Sigma_Infty}
defines the limiting smile as maturity tends to infinity.

\subsubsection{Heston model}
\label{subsec:Hseton}
The characteristics
$F,R$
are given in~\eqref{Eq:FR_Heston}
and 
$\chi(u)=u\zeta\rho-\lambda$
(see~\eqref{eq:Heston_chi}).
Assumption~A5 
is satisfied if and only if 
$\chi(1) < 0$, which is equivalent to  
$\lambda > \zeta \rho$. 
Since
$\lambda \theta \neq 0$ 
it follows that
$w \mapsto F(0,w)$ is not 
identically $0$.
Since the assumptions of 
Corollary~\ref{cor:LCGFh}~(ii)
are satisfied,
Theorem~\ref{thm:AsymVol}
implies that the limiting smile
is given by the formula in~\eqref{eq:Def_Sigma_Infty},
where
\begin{eqnarray}
\label{eq:h_Heston}
h(u) & = &  - \frac{\lambda \theta }{\zeta^2}\left(\chi(u) + \sqrt{\Delta(u)}\right)
\qquad\text{and}\quad 
\Delta(u)  = \chi(u)^2 - \zeta^2 (u^2 - u). 
\end{eqnarray}
This implies the main result in~\cite{FJ},~\cite{FJM_note}.
A first order asymptotic expansion for the large maturity smile in the 
Heston model was obtained in~\cite{FJM} using saddle point methods.


\subsubsection{Heston model with state-independent jumps}
The functions 
$F,R$
are given in~\eqref{Eq:HEston_Jumps}
and 
$\chi(u)=u\zeta\rho-\lambda$.
As in Section~\ref{subsec:Hseton},
$\lambda > \zeta \rho$
implies that 
$(X,V)$ 
defined Section~\ref{sec:Heston_Jumps} is a non-degenerate affine stochastic
volatility model
that satisfies~A5.
As before assumption 
$\lambda \theta \neq 0$
implies that 
$w\mapsto F(0,w)$
is non-zero.
$\widetilde{\kappa}(u)$,
defined in~\eqref{eq:Kappa_Def},
is a cumulant generating function of the compensated 
pure-jump L\'evy process 
$J$.
Assume that there
exists 
$\kappa_- < 0$ 
such that 
$|\widetilde{\kappa}(u)|<\infty$
for $u>\kappa_-$,
$|\widetilde{\kappa}(u)|=\infty$
for 
$u<\kappa_-$
and~\eqref{eq:Lambda_Steep}
holds for 
$u_-=\kappa_-$
and
$u_+=\infty$
(e.g. if the distribution of the absolute jump heights is exponential
with parameter
$\alpha>0$,
then
$\kappa_-=-\alpha$).
Under these assumptions on state-independent
jumps, the function 
$F$
in~\eqref{eq:F_HEston_Jumps}
is steep and 
$\{(0,0),(1,0)\}\subset\cD_F^\circ$.
Hence 
Theorem~\ref{thm:LCGFh}
implies that the limiting cumulant generating function is of the form
$$
h(u) = - \frac{\lambda \theta }{\zeta^2}\left(\chi(u) + \sqrt{\Delta(u)}\right) + \widetilde{\kappa}(u),
$$
where 
$\Delta(u)$
is as
in~\eqref{eq:h_Heston},
and
Theorem~\ref{thm:AsymVol}
yields the limiting smile 
formula in~\eqref{eq:Def_Sigma_Infty}.
Note also that
condition~\eqref{eq:Lambda_Steep}
on the jump measure
is not necessary if 
$\Delta(u)<0$
for some 
$u>\kappa_-$,
since in this case
$F$
in~\eqref{eq:F_HEston_Jumps}
is automatically
steep.



\subsubsection{A model of Bates with state-dependent jumps}
The functions 
$F,R$
are given in~\eqref{Eq:Bates}.
Again we assume 
$\lambda > \zeta \rho$
and 
$\lambda \theta \neq 0$,
which implies that 
$(X,V)$
defined in Section~\ref{subsec:Bates}
is a non-degenerate
affine stochastic
volatility model
that satisfies A5.
It is clear
from~\eqref{eq:F_Bates}
that the assumptions 
of Theorem~\ref{thm:LCGFh}
on 
$F$ 
are satisfied.
Let
$\widetilde{\kappa}(u)$
be as in~\eqref{eq:R_Bates}
and assume that
either
$|\widetilde{\kappa}(u)|<\infty$
for all
$u\in\bbR$
or there exists
$\kappa_-\in\bbR$
such that 
$|\widetilde{\kappa}(u)|<\infty$
for $u>\kappa_-$,
$|\widetilde{\kappa}(u)|=\infty$
for 
$u<\kappa_-$
and 
$\lim_{n\to\infty}|\widetilde{\kappa}(u_n)|=\infty$
for any sequence 
$(u_n)_{n\in\bbN}$
with
$u_n\downarrow\kappa_-$.
Then 
$\cD_R$
is open in
$\bbR^2$
(see~\eqref{eq:R_Bates})
and
Theorem~\ref{thm:LCGFh}
implies that
the limiting
cumulant generating function takes the form
$$
h(u) =-\frac{ \lambda\theta }{\zeta^2}\left(\chi(u)+\sqrt{\Delta(u)}\right) ,\qquad\text{where}\quad
\Delta(u) = \chi(u)^2 - \zeta^2 (u^2 - u + 2 \widetilde{\kappa}(u)).
$$


\subsubsection{The Barndorff-Nielsen-Shephard model}
The functions 
$F,R$
are given in~\eqref{Eq:BNS}.
Since $\chi(u)=-\lambda < 0$
and the jump measure is non-trivial
(i.e. $\nu\neq0$),
the process
$(X,V)$
defined in Section~\ref{subsec:BNS_Model}
is a non-degenerate affine stochastic volatility process.
Assume that $\kappa(u)$ 
defined in~\eqref{eq:Kappa_BNS}
is either finite for all 
$u\in\bbR$
or 
there exists 
$\kappa_+>0$
with 
$|\kappa(u)|<\infty$
for 
$u<\kappa_+$,
$|\kappa(u)|=\infty$
for 
$u>\kappa_+$
and
and~\eqref{eq:Lambda_Steep}
holds for 
$u_-=\kappa_-$
and
$u_+=\infty$
(e.g. if the distribution of the absolute jump heights is exponential
with parameter
$\alpha>0$,
then
$\kappa_-=-\alpha$).
Then 
$F$
(see~\eqref{eq:F_BNS})
satisfies the assumptions of Theorem~\ref{thm:LCGFh}
and the limiting cumulant generating function is of the form
$$
h(u) = \lambda \kappa \left(\frac{u^2}{2\lambda} + u\left(\rho - \frac{1}{2\lambda}\right)\right) -
u\lambda\kappa(\rho).
$$

\subsection{How close are the formula $\sigma_\infty(x)$ and the implied volatility $\sigma_t(x)$ for
large maturity?}
\label{subsec:HowClose}
In this section we plot the difference 
$|\sigma_\infty(x)-\sigma_t(x)|$
for 
$t\in\{10,15\}$
and 
$x\in[-0.1,0.1]$
for the models with jumps from Section~\ref{subsec:PlotsOfSmiles}
(see Figure~\ref{fig:Errors}).
In the case 
$t$
equals 
10 years
the error is approximately 
$45$
basis points (bp)
with the strike 
$K$
ranging from
$30\%$
to 
$200\%$
of the spot.
At the maturity of 15 years
the error
is approximately 
$20$ bp
and
$K$
ranges between 
$20\%$
and
$400\%$
of the spot.

In the cases of 
Heston with state-independent jumps and Bates with state-dependent jumps
we took the following diffusion parameters 
\begin{equation*}
\lambda=1.15,\quad
\zeta=0.2,\quad
\theta=0.04,\quad
\rho=-0.4
\end{equation*}
and the L\'evy measure
$\nu\left(\dd \xi_1\right)=\alpha\E^{\alpha\xi_1}I_{\left\{\xi_1<0\right\}}\dd\xi_1$
with 
$\alpha=0.6$.
The compensated cumulant generating function 
$\wt \kappa(u)$
(see~\eqref{eq:Kappa_Def}
for the definition of 
$\wt \kappa(u)$
in Section~\ref{sec:Heston_Jumps} 
and note that it takes the same form in
Section~\ref{subsec:Bates})
is in this case given by
$$\widetilde{\kappa}(u)= \frac{u\left(u-1\right)}{\left(u+\alpha\right)\left(\alpha+1\right)}
\qquad\text{for all }u\in\left(-\alpha,\infty\right).$$

In the case of the BNS model we took 
a pure-jump subordinator 
$J$
with L\'evy measure
$\nu\left(\dd \xi_2\right)=ab\E^{-b\xi_2}I_{\left\{\xi_2>0\right\}}\dd\xi_2$.
The cumulant generating function~\eqref{eq:Kappa_BNS}
is
given by 
$\kappa(u)=a u/(b-u)$
for 
$u<b$.
We used the following values for the parameters
\begin{equation*}
a=1.4338,\quad
b=11.6641,\quad
\lambda=0.5783,\quad
\rho=-1.2606,
\end{equation*}
which were
taken from~\cite[Section 7.3]{Schoutens} 
where the model was calibrated to the options on the S\&P~500. 


\begin{figure}
\begin{center}
\subfigure{\includegraphics[scale=0.35]{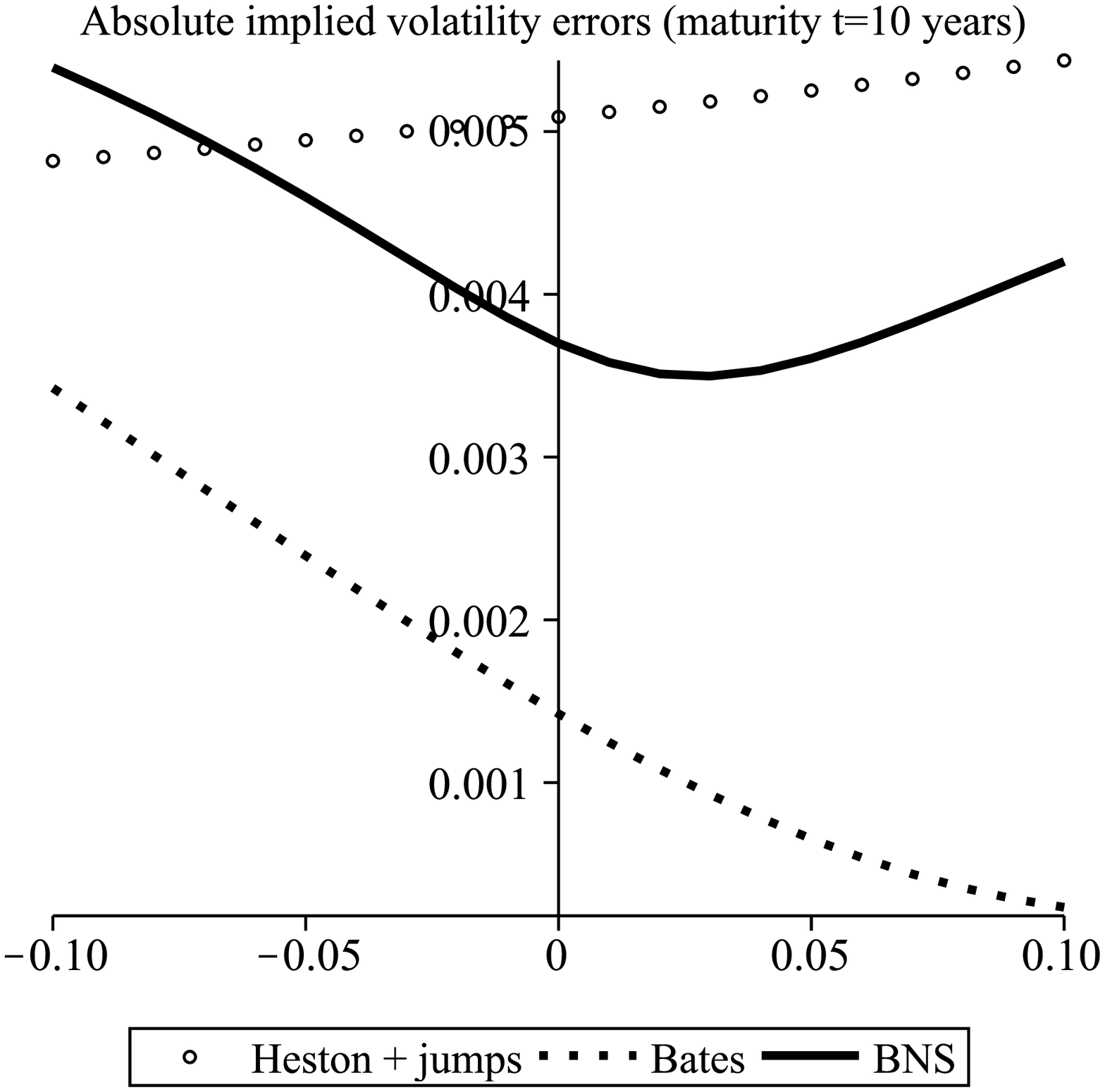}}
\subfigure{\includegraphics[scale=0.35]{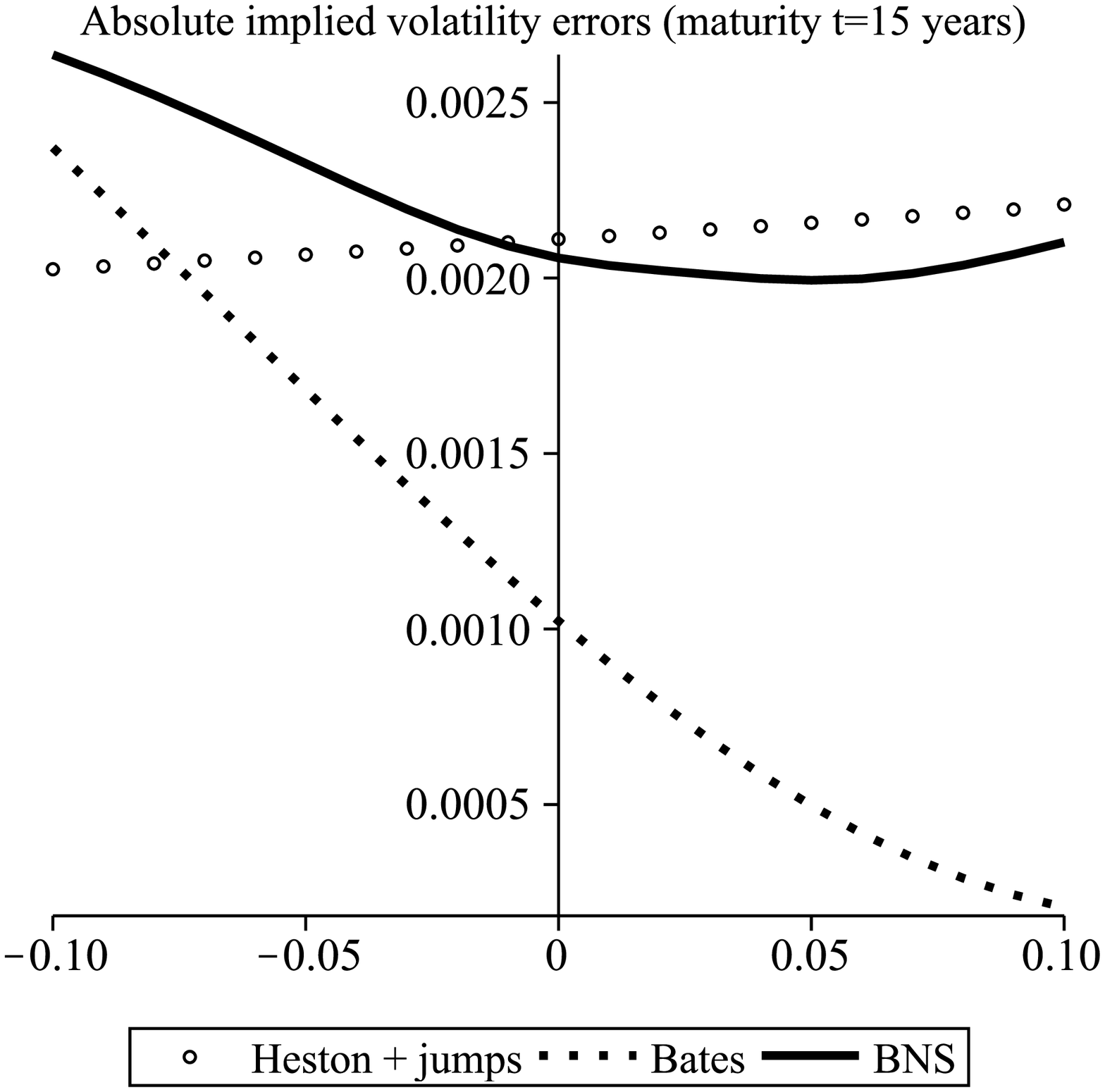}}
\caption{\footnotesize This figure contains the plots of the function
$x\mapsto |\sigma_\infty(x)-\sigma_t(x)|$
in the interval
$x\in[-0.1,0.1]$
for the models with jumps from Section~\ref{subsec:PlotsOfSmiles}
and maturities 
$t\in\{10,15\}$.
The values of the model parameters used are given 
in Section~\ref{subsec:HowClose}.}
\label{fig:Errors}
\end{center}
\end{figure}



\bibliographystyle{alpha}
\bibliography{references}

\end{document}